\documentclass[12pt,draftclsnofoot,onecolumn,nofoot]{IEEEtran}
\usepackage{graphics}
\usepackage{graphicx}
\usepackage{epsfig}
\usepackage{amsmath,amssymb,amsfonts,amstext,amsbsy,amsopn,eucal,dsfont}
\usepackage{sublabel}

\newtheorem{theorem}{\textbf{Theorem}}
\newtheorem{proposition}{\textbf{Proposition}}
\newtheorem{lemma}{\textbf{Lemma}}
\newtheorem{remark}{\textbf{Remark}}
\newtheorem{corollary}{\textbf{Corollary}}
\newtheorem{algorithm}{\textbf{Algorithm}}

\newcommand{\etal}{\emph{et al.}}
\newcommand{\defn}{\triangleq}
\newcommand{\ie}{\emph{i.e.}, }
\newcommand{\dotleq}{\overset{\cdot}{\leq}}

\begin{document}
\title{Network Coding with Two-Way Relaying: Achievable Rate Regions and Diversity-Multiplexing Tradeoffs}

\author{Chun-Hung Liu, Feng Xue and Jeffrey G. Andrews
\thanks{Liu and Andrews are with the Department of Electrical and Computer Engineering, the University of Texas at Austin, Austin TX 78712-0204, USA. Xue is with Intel Corp. at Santa Clara CA, USA. Parts of this work were presented at the IEEE International conference on communications, Beijing China, May 2008 and at the Allerton conference on communication, control and computing, IL USA, Sep. 2007. This work was supported in part by the National Science Council of Taiwan, R.O.C. under Contract TMS-94-1A-050. The contact author is J. G. Andrews (Email: jandrews@ece.utexas.edu). Manuscript date: \today.}}

\maketitle

\begin{abstract}
This paper addresses the fundamental characteristics of information exchange via multihop network coding over two-way relaying in a wireless ad hoc network. The end-to-end rate regions achieved by time-division multihop (TDMH), MAC-layer network coding (MLNC) and PHY-layer network coding (PLNC) are first characterized. It is shown that MLNC does not always achieve better rates than TDMH, time sharing between TDMH and MLNC is able to achieve a larger rate region, and PLNC dominates the rate regions achieved by TDMH and MLNC. An opportunistic scheduling algorithm for MLNC and PLNC is then proposed to stabilize the two-way relaying system for Poisson arrivals whenever the rate pair is within the Shannon rate regions of MLNC and PLNC. To understand the two-way transmission limits of multihop network coding, the sum-rate optimization with or without certain traffic pattern and the end-to-end diversity-multiplexing tradeoffs (DMTs) of two-way transmission over multiple relay nodes are also analyzed.
\end{abstract}


\section{Introduction}
Users in a multihop wireless network convey information to each other with the help of intermediary \emph{routing} nodes. Generally, packets are forwarded by the relays towards their respective destinations in a \emph{decode-and-forward} fashion. After the seminal paper by Ahlswede \etal \cite{RANCSYRLRWY00} on network coding for wireline networks, it is known that better performance is possible if intermediate nodes are allowed to change the content of their packets \cite{CFES0701,CFES0702}. Wireless network coding can improve throughput and reliability due to the broadcast nature of the wireless medium, and the resulting opportunities to gather information from all audible transmissions \cite{SKHRWHDKMMJC08,DNTNBB99,PLNJKES06,CHLFXSS08,FXCHLSS07,PPHY0107}. Since two-way traffic is inherent to peer-to-peer communication, the two-way relaying channel is a key building block for information exchange over multiple hops.  This paper provides fundamental characterization and limits of the two-way relaying channel with (and without) \emph{multihop} wireless network coding.

Specifically, we consider two multihop network coding protocols as illustrated in Fig. \ref{Fig:TwoWayRelaySys}(a) in this work, \ie MAC-Layer network coding (MLNC) and PHY-Layer network coding (PLNC). MLNC is a coding operation that happens at or above the media access layer \cite{RWYSYRLNCZZ06,SYRLRWYNC03,CHJH06,SKHRWHDKMMJC08,FXCHLSS07,CHLFXSS08}. PLNC is based on the coding protocol proposed in \cite{FXSS07,CSTJOSS07,GKSS07}. In order to perform MLNC and PLNC, the two source nodes respectively transmit their packets to the relay node in the first two time slots. Then the relay node constructs a new network-coded packet from the received packets and broadcasts it to the two source nodes in the third time slot. If MLNC is adopted -- whereby the packets are manipulated before channel coding -- the transmission rate of the relay node is limited by the smaller channel capacity \cite{FXSS07}. Since the network coding procedure of PLNC is performed instead on channel coded data, PLNC can individually achieve the broadcast channel capacities from the relay node to the source/destination nodes \cite{SJKPMVT08, CSTJOSS07, GKSS07, FXSS07,FXCHLSS07}. Therefore, PLNC always achieves better throughput than MLNC, especially when the two channels are highly asymmetric. PLNC's performance gains come at the cost of transceiver complexity.

\subsection{Motivation and Related Work}
Traditionally (pre-network coding), information exchange between two users via a relay has been accomplished by a time-division multihop (TDMH) protocol in four time slots\footnote{Frequency-division could be used as well. In this paper the comparisons focus on time-division systems only, and comparisons can be applied to frequency-division systems similarly.}, as shown in Fig. \ref{Fig:TwoWayRelaySys}(a). Intuitively, the MLNC and PLNC protocols {\em save} one time slot.  In \cite{PLNJKES06}, two-way relaying for cellular systems was considered, while \cite{SKHRWHDKMMJC08} and \cite{SKDKWHMMJC0905} proposed an MLNC algorithm effective for wireless mesh networks in heavy traffic. The network coding protocol in Fig. \ref{Fig:TwoWayRelaySys}(b) can be further reduced to two slots if advanced joint coding/decoding -- i.e. analog network coding (ANC) -- is allowed. In this case, both source nodes send their packets to the relay node simultaneously during the first slot; then the relay node either amplifies and broadcasts the signals, or broadcasts the XOR-ed packets after decoding by successive interference cancellation \cite{SKSGDK07,PPHY0607,PPHY0107,MCAY08}. The achievable rates for analog network coding were studied in \cite{SKIMAGDKMM07,PPHY0107}.

Although MLNC and analog network coding have been shown to achieve throughput gains in certain environments\cite{PLNJKES06,PPHY0607,CHLFXSS08}, it is unclear how channel realizations influence the achievable rates and whether MLNC and/or PLNC are always better than TDMH in terms of end-to-end throughput. In terms of implementation complexity, the four techniques can be ranked from low to high as TDMH, MLNC, PLNC ANC. Analog network coding requires stringent synchronization and joint decoding and could suffer from decoding error propagation due to channel estimation and quantization errors.  Characterizing the exact achievable rate regions of TDMH, MLNC and PLNC will provide guidance on the tradeoffs between them.

A frequently neglected but very important consideration for network coding is that the traffic patterns can vary significantly, and have a significant effect on the ability to achieve the gains promised by network coding.   For example, data downloads are essentially one-way traffic, while peer-to-peer conversations are fairly symmetric. Analyzing the gains from network coding in a two-way relay network in view of the traffic pattern allows insight into how such systems should be designed.  In this paper we explicitly consider the traffic pattern in our results.

Finally, network coding can be used to exploit cooperative diversity between source and destination nodes\cite{LXTEFJKDJC07,YCSKJL06}. Since network coding is able to provide diversity as well as throughput gain, it is of interest to understand the diversity-multiplexing tradeoffs (DMTs)\footnote{The diversity-multiplexing tradeoff (DMT) for point-to-point multiple input and multiple output (MIMO) channels was found in \cite{LZDNCT03}, and has become a popular metric for comparing transmission protocols.} of MLNC and PLNC and determine if they are better than TDMH's. Since we consider two-way transmission over multiple relays, this plurality of relays may cooperate in a number of different ways or not at all, and each cooperation scenario leads to a different DMT result for TDMH, MLNC and PLNC.

\subsection{Contributions}
In this paper, we first characterize the exact achievable rate regions of TDMH, MLNC and PLNC and show that MLNC is not always superior to TDMH, while PLNC has the largest rate region among the three. An opportunistic scheduling algorithm is then presented to achieve larger rate region than those achieved independently by TDMH and MLNC. The stability of this scheduling under Poisson arrivals is shown via queuing analysis. Subsequently, the sum rates of the three protocols with or without traffic constraint are determined. The optimal sum rate with or without a traffic pattern constraint -- defined as the ratio between the rates of the two directions -- is also characterized.

The DMTs of the three transmission protocols with different relay collaboration scenarios are all derived. They are quite different from previous multihop DMT results (e.g. see \cite{DGAGHVP08,RVRWH08,ABAKDPRAL06,JNLGWW03}) due to their dependence on the traffic pattern, time allocation to each transmission direction, and the number of cooperative relay nodes in the system. We first characterize the DMT of TDMH with relay collaboration and show that it coincides to that of TDMH with an optimally selected relay node. Then, MLNC and PLNC are shown to achieve the same DMT in the high signal-to-noise ratio (SNR) regime, and that better tradeoffs are attained if an optimal relay node is selected to broadcast. Considering these results, we conclude that selecting an optimal relay node to receive and transmit (or broadcast) is preferable. Finally, we find that MLNC and PLNC can have a worse DMT if there is suboptimal time allocation for a certain traffic pattern. Intuitively, if the offered traffic load is much higher in the forward direction than the backward direction relative to one of the source nodes, then network coding may not be helpful for that source since it presumes a symmetric data rate in the high SNR regime.



\section{System Model and Assumptions}\label{Sec:SysmMode&Assu}
Consider a multihop wireless network in which information exchange by multihop routing can be characterized by considering the two-way relaying system illustrated in Fig. \ref{Fig:TwoWayRelaySys}, where the two source nodes A and B would like to exchange their data packets $W_A$ and $W_B$, denoted as two binary sequences. In the sequel, we assume there is no direct channel between the two source nodes, otherwise, mutihop is not needed \cite{OOSS06,MSJNLMHDJCTF06}. All nodes in the network are assumed to be single-antenna and \emph{half-duplex}, \ie no nodes can transmit and receive at the same time. Also, we assume that channel gains are constant during the whole transmission, and the channel gain between any two nodes X and Y, denoted by  $h_{XY}$, is reciprocal and modeled as a zero-mean, independent, circularly symmetric complex Gaussian random variable with variance $1/\sigma_{XY}$.

The core idea of MLNC is that the relay constructs a new network-coded packet $W_D$ after receiving $W_A$ and $W_B$, and then it broadcasts $W_D$ to both source nodes. $W_D$ is obtained by directly XOR-ing $W_A$ and $W_B$ bitwise, \ie $W_D = W_A \oplus W_B$. As a source node receives $W_D$, it can decode the new packet by XOR-ing $W_D$ with its side information, \ie $W_A$ or $W_B$. In contrast to XOR-ing data contents in MLNC, PLNC does network coding on channel codes, i.e., after channel encoding to individual data contents \cite{FXSS07,CSTJOSS07,GKSS07}. Taking binary symmetric channel for example, the relay broadcasts $X_{AB}\oplus X_{BA}$ after receiving $W_A$ and $W_B$, where $X_{AB}$ and $X_{BA}$ stand for the channel codes of $W_A$ and $W_B$ respectively, as shown in Fig. \ref{Fig:BinaryPLNC}. The channel encoder of $W_A$ ($W_B$) which generates $X_{AB}$ ($X_{BA}$) is designed according to the channel condition from relay node D to source node A (B). So when the XOR-ed channel code is received by a source node, it can subtract $X_{BA}$ or $X_{AB}$ before channel decoding, as shown in Fig. \ref{Fig:BinaryPLNC}. As shown in \cite{FXSS07,CSTJOSS07,GKSS07}, as long as there exists a common input distribution at the relay node that achieves channel capacities for both directions (e.g., binary symmetric channel and AWGN channel), similar idea works. In this paper, we assume the existence of such common input distribution. Due to the different constructions, the available broadcast rates for MLNC and PLNC are different. For MLNC, the broadcast rate is limited by the smaller broadcast capacity, due to the fact that both ends need decode the same (XOR-ed) message. For PLNC, each direction can achieve its individual channel capacity \cite{FXSS07}. Let $C_{XY}$ denote the channel capacity from nodes X to Y. The broadcast rates of MLNC and PLNC can be concluded from \cite{FXSS07}\cite{GKSS07} as follows.
\begin{lemma}\label{Lem:BCRateNC}
The achievable broadcast rate of MLNC for the relay node is $C_{\min}$ for both directions, where $C_{\min}\defn \min\{C_{DA},C_{DB}\}$. If there exists a common input distribution that maximizes the mutual information from relay node $D$ to nodes $A$ and $B$, then the achievable broadcast rates of PLNC for both directions are respectively $C_{DA}$ and $C_{DB}$.
\end{lemma}

In this paper, we also consider the information exchange between two source nodes can be completed with multiple relay nodes, as shown in Fig. \ref{Fig:TwoWayRelaySys}(b). $\mathcal{D}_{AB}$ denotes the set of the relay nodes available\footnote{where ``available'' means any relay node in $\mathcal{D}_{AB}$ can successfully decode the information from both source nodes.} between source nodes $A$ and $B$. Denote by $|\mathcal{D}_{AB}|$ the number of relay nodes in $\mathcal{D}_{AB}$, which is usually a random variable for different time slots; however, to facilitate the analysis here we assume it remains constant during the period of exchanging packets. We assume all nodes in $\mathcal{D}_{AB}$ are close and able to collaborate under reasonable communication overhead so that every relay node can share its received information with other relays. In this context, $\mathcal{D}_{AB}$ \emph{virtually} becomes a big relay node equipped with $|\mathcal{D}_{AB}|$ antennas. The channels from node A to $\mathcal{D}_{AB}$ become a single-input-multiple-output (SIMO) channel (or a MISO channel from $\mathcal{D}_{AB}$ to node A). Therefore, receive maximum ratio combining (MRC) and transmit MRC can be accomplished in $\mathcal{D}_{AB}$ for TDMH assuming joint processing can be carried out. Although receive MRC can be performed in the first two transmission stages for MLNC and PLNC, it is hard to achieve bidirectional transmit MRC in the broadcast stage. Thus in analyzing MLNC and PLNC with relay collaboration we  consider two scenarios of broadcasting, \ie all relay nodes broadcast at the same time and only an optimally selected relay node broadcasts.

\section{Achievable Rate Region, Opportunistic Network Coding and Scheduling} \label{Sec:AchiRateRegi&RelaSeleAlgo}
In this section, we are interested in determining the \emph{end-to-end} rate pair $(R_{AB},R_{BA})$ achieved by the aforementioned three protocols. For convenience, we call $R_{AB}$ the forward rate, $R_{BA}$ the backward rate and $\mu$ the traffic pattern parameter which is the ratio $R_{AB}/R_{BA}$. Here we only characterize the end-to-end rate regions achieved by TDMH, MLNC and PLNC for the single relay network in Fig. \ref{Fig:TwoWayRelaySys}(a) since they are easily extended to the multiple relay case. Given the achievable rate regions, two opportunistic packet scheduling algorithms are  respectively proposed for MLNC and PLNC, and their stability with random arrivals are characterized as well.

\subsection{Achievable Rate Regions for Two-Way Transmission Protocols over a Signal Relay}\label{SubSec:AchiRateRegi}
For the two-way relaying system in Fig. \ref{Fig:TwoWayRelaySys}, we assume that $C_{DA}$ and $C_{DB}$ are achieved by the same input distribution. The achievable rate region is basically constructed by the forward and backward rate pairs $(R_{AB},R_{BA})$. First consider TDMH that needs four time slots to exchange packets. Since its achievable Shannon rate pairs are constrained by time allocations in the four time slots, its achievable rate region is
\begin{eqnarray}\label{Eqn:AchiRateRegiTDMH}
\mathcal{R}_{\texttt{TDMH}} \defn \bigg\{(R_{AB},R_{BA}): R_{AB}\leq \{\lambda_1C_{AD},\lambda_2C_{DB}\},
R_{BA}\leq\{\lambda_3C_{BD},\lambda_4C_{DA}\}, \sum_{k=1}^4 \lambda_k =1 \bigg\},
\end{eqnarray}
where $\{\lambda_k\}\in [0,1]$ are the time-allocation parameters for four transmission time slots. Define $\Sigma_{AB}\defn(1/C_{AD}+1/C_{DB})^{-1}$ and $\Sigma_{BA}\defn(1/C_{BD}+1/C_{DA})^{-1}$. Therefore, we have the following:
\begin{theorem}\label{Thm:AchieRateRegiTDMH}\emph{
$\mathcal{R}_{\texttt{TDMH} }$ is the triangle with vertices $\mathbf{0}$, $(\Sigma_{AB},0)$ and $(0,\Sigma_{BA})$, as shown in Fig. \ref{Fig:AchiRateRegi}.}
\end{theorem}
\begin{proof}
Vertex $(0,\Sigma_{BA})$ corresponds to the case of one-way backward traffic, it is achieved by setting $\lambda_1=\lambda_2=0$, $\lambda_3=\frac{\Sigma_{BA}}{C_{BD}}$ and $\lambda_4=\frac{\Sigma_{BA}}{C_{DA}}$. Similarly, vertex $(\Sigma_{AB},0)$ corresponds to the case of one-way forward traffic, achieved by setting $\lambda_3=\lambda_4=0$, $\lambda_1=\frac{\Sigma_{AB}}{C_{AD}}$ and $\lambda_2=\frac{\Sigma_{AB}}{C_{DB}}$. Since $\mathcal{R}_{\texttt{TDMH} }$ is described by linear constraints, it is convex and thus achievable. Now we show that $\mathcal{R}_{\texttt{TDMH} }$ is also an outer bound for TDMH. Consider the four linear constraints of transmission rates in \eqref{Eqn:AchiRateRegiTDMH}. Dividing each of them by their corresponding channel capacity and adding them up, we obtain
$$\frac{R_{AB}}{C_{AD}}+\frac{R_{AB}}{C_{DB}}+\frac{R_{BA}}{C_{BD}}+\frac{R_{BA}}{C_{DA}}=\frac{R_{AB}}{\Sigma_{AB}}+\frac{R_{BA}}{\Sigma_{BA}}\leq \sum_{k=1}^4 \lambda_k=1.$$
This is exactly the region below the line connecting vertices $(0,\Sigma_{BA})$ and $(\Sigma_{AB},0)$.
\end{proof}

For MLNC, by Lemma \ref{Lem:BCRateNC} its achievable rates are defined in similar fashion as in \eqref{Eqn:AchiRateRegiTDMH} as follows
\begin{eqnarray}\label{Eqn:AchiRateRegiMLNC}
  \mathcal{R}_{\texttt{MLNC} }\defn \bigg\{(R_{AB},R_{BA}): R_{AB}\leq \{\lambda_1C_{AD},\lambda_3C_{\min}\}, R_{BA}\leq \{\lambda_2C_{BD},\lambda_3C_{\min}\},\sum_{k=1}^3\lambda_k=1 \bigg\}.
\end{eqnarray}
Define $\Sigma_{AA}\defn(1/C_{AD}+1/C_{DA})^{-1}$, $\Sigma_{BB}\defn(1/C_{BD}+1/C_{DB})^{-1}$ and $\Sigma_{ABB}\defn(1/\Sigma_{AB}+1/C_{BD})^{-1}$. Then we have the following theorem characterizing the achievable rate region of MLNC.
\begin{theorem}\label{Thm:AchieRateRegiMLNC}
\emph{If $C_{DB}\leq C_{DA}$, then $\mathcal{R}_{\texttt{MLNC} }$ is the quadrilateral with vertices $\mathbf{0}$, $(\Sigma_{AB},0)$,  $(0,\Sigma_{BB})$ and $(\Sigma_{ABB},\Sigma_{ABB})$, as shown in Fig. \ref{Fig:AchiRateRegi}(a). If $C_{DB} > C_{DA}$, then $\mathcal{R}_{\texttt{MLNC} }$ is the quadrilateral with vertices $\mathbf{0}$, $(\Sigma_{AA},0)$, $(\Sigma_{ABB},\Sigma_{ABB})$ and $(0,\Sigma_{BA})$, as shown in Fig. \ref{Fig:AchiRateRegi}(b).}
\end{theorem}
\begin{proof}
First consider the case $C_{DB}\leq C_{DA}$. The achievable rate region $\mathcal{R}_{\texttt{MLNC} }$ in \eqref{Eqn:AchiRateRegiMLNC} becomes
\begin{eqnarray}\label{Eqn:AchiRateRegi01}
 \mathcal{R}_{\texttt{MLNC} }= \bigg\{(R_{AB},R_{BA}): R_{AB}\leq \{\lambda_1C_{AD},\lambda_3C_{DB}\}, R_{BA}\leq \{\lambda_2C_{BD},\lambda_3C_{DB}\}, \sum_{k=1}^3 \lambda_k =1\bigg\}.
\end{eqnarray}
Vertex $(\Sigma_{AB},0)$ corresponds to the case of one-way forward traffic, and is achieved by setting $\lambda_1=\frac{\Sigma_{AB}}{C_{AD}}$, $\lambda_2=0$ and $\lambda_3=\frac{\Sigma_{AB}}{C_{DB}}$. Vertex $(0,\Sigma_{BB})$ corresponds to the case of one-way backward traffic, and is achieved by setting $\lambda_1=0$, $\lambda_2=\frac{\Sigma_{2}}{C_{BD}}$ and $\lambda_3=\frac{\Sigma_{BB}}{C_{DB}}$. Finally vertex $(\Sigma_{ABB},\Sigma_{ABB})$ is achieved by setting $\lambda_1=\frac{\Sigma_{ABB}}{C_{AD}}$, $\lambda_2=\frac{\Sigma_{ABB}}{C_{BD}}$ and $\lambda_3=\frac{\Sigma_{ABB}}{C_{DB}}$. Because \eqref{Eqn:AchiRateRegi01} is a convex quadrilateral defined by the three vertices and $(0,0)$, the region \eqref{Eqn:AchiRateRegi01} is achievable by time-sharing among the four vertices. Next, by time-sharing we show the quadrilateral $\mathcal{R}_{\texttt{MLNC} }$ in \eqref{Eqn:AchiRateRegi01} is also an outer bound for MLNC. Consider the four linear constraints in \eqref{Eqn:AchiRateRegi01} and divide each of them by their corresponding channel capacity. We have
\sublabon{equation}
\begin{eqnarray}
  \frac{R_{AB}}{C_{AD}}+\frac{R_{AB}}{C_{DB}}+\frac{R_{BA}}{C_{BD}}=\frac{R_{AB}}{\Sigma_{AB}}+ \frac{R_{BA}}{C_{BD}} &\leq& \sum_{k=1}^3\lambda_k=1, \label{Eqn:ConsMLNC01}\\
  \frac{R_{BA}}{C_{BD}}+\frac{R_{BA}}{C_{DB}}+\frac{R_{AB}}{C_{AD}}=\frac{R_{BA}}{\Sigma_{BB}}+\frac{R_{AB}}{C_{AD}} &\leq& \sum_{k=1}^3\lambda_k=1. \label{Eqn:ConsMLNC02}
\end{eqnarray}
\sublaboff{equation}
Equation \eqref{Eqn:ConsMLNC01} corresponds to the region below the line connecting $(\Sigma_{AB},0)$ and $(\Sigma_{ABB},\Sigma_{ABB})$, and \eqref{Eqn:ConsMLNC02} corresponds to the region below the line connecting $(\Sigma_{ABB},\Sigma_{ABB})$ and $(0,\Sigma_{BB})$. This completes the proof for the case $C_{DB}\leq C_{DA}$. The proof for the case $C_{DB}>C_{DA}$ is similar.
\end{proof}

\begin{remark}
According to Theorems \ref{Thm:AchieRateRegiTDMH}, \ref{Thm:AchieRateRegiMLNC} and Fig. \ref{Fig:AchiRateRegi}(a)(b), MLNC is not always better than TDMH. For example in Fig. \ref{Fig:AchiRateRegi}(a), MLNC is worse than TDMH when $C_{DB}<C_{DA}$ and $\mu$ is greater than certain value. This is because the broadcast rate of MLNC is limited by the worse broadcast channel, as stated in Lemma \ref{Lem:BCRateNC}. A hybrid protocol of time sharing between MLNC and TDMH, as indicated in Fig. \ref{Fig:AchiRateRegi}(a)(b) by a dashed line, can achieve a larger rate region $\mathbf{CovH}(\mathcal{R}_{\texttt{TDMH}},\mathcal{R}_{\texttt{MLNC}})$, the convex hull of $\mathcal{R}_{\texttt{TDMH} }$ and $\mathcal{R}_{\texttt{MLNC} }$. Thus, in practice the relay can decide to use TDMH or MLNC according to the results in Fig. \ref{Fig:AchiRateRegi}.
\end{remark}

For PLNC, its achievable rate region is constructed as follows according to Lemma \ref{Lem:BCRateNC}:
\begin{eqnarray}\label{Eqn:AchiRateRegiPLNC}
 \mathcal{R}_{\texttt{PLNC} }\defn\bigg\{(R_{AB},R_{BA}) : R_{AB}\leq \{\lambda_1C_{AD},\lambda_3C_{DB}\},
 R_{BA}\leq \{\lambda_2C_{BD},\lambda_3C_{DA}\}, \sum_{k=1}^3 \lambda_k=1\bigg\}.
\end{eqnarray}
Define $\Sigma_{ABA}\defn(1/\Sigma_{AB}+C_{DA}/C_{BD}/C_{DB})^{-1}$ and $\Sigma_{BAB}\defn(1/\Sigma_{BA}+C_{DB}/C_{AD}/C_{DA})^{-1}$. Then we have the following theorem.
\begin{theorem}\label{Thm:AchieRateRegiPLNC}
\emph{$\mathcal{R}_{\texttt{PLNC} }$ is the quadrilateral with vertices $\mathbf{0}$, $(\Sigma_{AB},0)$, $(\Sigma_{ABA},\Sigma_{BAB})$ and $(0,\Sigma_{BA})$.}
\end{theorem}
\begin{proof}
The proof is omitted here since it is similar to the proof of Theorem \ref{Thm:AchieRateRegiMLNC}.
\end{proof}

\begin{remark}\label{Rem:AchRatRegiPLNC}
Given the above results in Theorems \ref{Thm:AchieRateRegiTDMH}-\ref{Thm:AchieRateRegiPLNC} and illustrated in Fig. \ref{Fig:AchiRateRegi}, we know the achievable rate region of PLNC is larger than $\mathbf{CovH}(\mathcal{R}_{\texttt{TDMH}},\mathcal{R}_{\texttt{MLNC}})$ as long as $C_{DA}\neq C_{DB}$. Also, it should be noticed that time sharing between TDMH and PLNC does not help achieve larger rate region as in the case of MLNC and TDMH.
\end{remark}

\subsection{Opportunistic Network Coding and Scheduling}\label{SubSec:OppoNetwCodi}
The achievable rate regions for the three transmission protocols have been characterized in Section \ref{SubSec:AchiRateRegi}. In this subsection we investigate the following question: If packets arrive at the two source nodes according to a random process, how should the packets be scheduled for transmission to maximize the rate region in which the queues are stable? We assume that the buffer size for queuing packets at the two source nodes is infinite. Packets of length $\ell$ (bits) arrive at the queue of source node A according to a Poisson process with rate $\acute{R}_{AB}$, while packets of length $\ell$ (bits) arrive at the queue of source node B according to another independent Poisson process with rate $\acute{R}_{BA}$. We also assume that $\ell$ is large and channel coding is perfect.

We first consider the case where the relay node uses MLNC to route the packets. Consider the time right after the previous transmission is complete. Let $q_A$ and $q_B$ denote the number of packets in the queues at node A and B, respectively. We propose the following opportunistic packet scheduling algorithm for MLNC.
\begin{algorithm}
\label{Alg:OMLNC} \centering\small
\textbf{\emph{Opportunistic MLNC and Scheduling}}
\begin{tabular}{|lp{4.0in}|}\hline
\textsf{Step 1} ($q_A \cdot q_B \neq 0$): & Two packets (one from each queue) are sent over Relay D by MLNC.\\
\textsf{Step 2} ($q_A\neq 0$, $q_B=0$): &  One packet from node A is sent over Relay D by TDMH.\\
\textsf{Step 3} ($q_A=0$, $q_B\neq 0$): & One packet from node B is sent over Relay D by TDMH. \\ \hline
\end{tabular}
\end{algorithm}
\vspace{0.1in}

It will be shown later in Section \ref{SubSec:SumRateOpti} that symmetric traffic through the relay node achieves the best network coding gain for MLNC. So when both buffers have packets, Algorithm \ref{Alg:OMLNC} schedules symmetric traffic; when one queue is empty certainly one-way traffic should be scheduled. By doing this, Algorithm \ref{Alg:OMLNC} is able to \emph{opportunistically} achieve $\mathbf{CovH}(\mathcal{R}_{\texttt{TDMH}},\mathcal{R}_{\texttt{MLNC}})$ by time sharing between TDMH and MLNC. Moreover, the attractive feature of this algorithm is that it does not require any arrival rate information of the two queues to achieve the system stability as stated in the following theorem.
\begin{theorem}\label{Thm:StabOMLNC}
\emph{Algorithm \ref{Alg:OMLNC} stabilizes the two-way relaying system for any Poisson arrivals if the (bit-arrival) rate pair $(R_{AB},R_{BA})\in\mathbf{CovH}(\mathcal{R}_{\texttt{TDMH}},\mathcal{R}_{\texttt{MLNC}})$ where $R_{AB}=\acute{R}_{AB}\ell$ and $R_{BA}=\acute{R}_{BA}\ell$.}
\end{theorem}
\begin{proof}
Without loss of generality, assume each packet has a unit length so that  $R_{AB}=\acute{R}_{AB}$ and $R_{BA}=\acute{R}_{BA}$. Now consider the time right after the \emph{n}-th transmission. Note that each transmission is one of the three steps in Algorithm \ref{Alg:OMLNC}. Denote the queue length vector at time $n$ as $\mathbf{Q}(n)\defn [q_A(n)\,\, q_B(n)]$. It is easy to verify that $\mathbf{Q}(n)$ forms a non-reducible Markov chain. Note that the number of packets arriving at node A during a transmission time slot $\Delta t$ is a Poisson process with parameter $R_{AB}\Delta t$, while the number of packets arriving at node B is also a Poisson process with parameter $R_{BA}\Delta t$.

In order to show the stability of the Markov process $\mathbf{Q}(n)$, we define the Lyapunov function as follows
\begin{equation}\label{Eqn:LyapFuncOMLNC}
   V(n)\defn \frac{\Sigma_{\min}}{\Sigma_{AB}-\Sigma_{\min}}q_A^2(n)+\frac{\Sigma_{\min}}{\Sigma_{BA}-\Sigma_{\min}}q_B^2(n)+2q_A(n)q_B(n),
\end{equation}
where $\Sigma_{\min}\defn (1/C_{AD}+1/C_{BD}+1/C_{\min})^{-1}$. Now consider the case in \textsf{Step 1} of Algorithm \ref{Alg:OMLNC}. In this case we have $q_A(n+1)=q_A(n)-1+\delta_A(n)$ and $q_B(n+1)=q_B(n)-1+\delta_B(n)$, where $\delta_A(n)$ and $\delta_B(n)$ are Poisson random variables with parameters $R_{AB}/\Sigma_{\min}$ and $R_{BA}/\Sigma_{\min}$, respectively. It follows that
\begin{eqnarray}\label{Eqn:ExpeLyapFuncOMLNC00}
  \mathbb{E}\left[V(n+1)|\mathbf{Q}(n)\right] &=& V(n)+2\frac{\Sigma_{AB}}{\Sigma_{\min}}\left[\frac{\Sigma_{\min}}{\Sigma_{AB}-\Sigma_{\min}}\left(\frac{R_{AB}}{\Sigma_{AB}}-1\right)+\frac{R_{BA}}{\Sigma_{AB}}\right]q_A(n) \nonumber\\
  &&+2\frac{\Sigma_{BA}}{\Sigma_{\min}}\left[\frac{\Sigma_{\min}}{\Sigma_{BA}-\Sigma_{\min}}\left(\frac{R_{BA}}{\Sigma_{BA}}-1\right)+\frac{R_{AB}}{\Sigma_{BA}}\right]q_B(n)+\Delta_a,
\end{eqnarray}
where $\Delta_a$ is a constant depending on $R_{AB},R_{BA},\Sigma_{\min}, \Sigma_{AB}$ and $\Sigma_{BA}$. Next consider the case in \textsf{Step 2} of Algorithm \ref{Alg:OMLNC}. We have $q_A(n+1)=q_A(n)-1+\hat{\delta}_A(n)$ and $q_B(n+1)=\hat{\delta}_B(n)$, where $\hat{\delta}_A(n)$ and $\hat{\delta}_B(n)$ are Poisson random variables with parameters $\frac{R_{AB}}{\Sigma_{AB}}$ and $\frac{R_{BA}}{\Sigma_{AB}}$, respectively. Thus,
\begin{equation} \label{Eqn:ExpeLyapFuncOMLNC01}
 \mathbb{E}\left[V(n+1)|\mathbf{Q}(n)\right] = V(n)+2\left[\frac{R_{BA}}{\Sigma_{AB}}+\frac{\Sigma_{\min}}{\Sigma_{AB}-\Sigma_{\min}}\left(\frac{R_{AB}}{\Sigma_{AB}}-1\right)\right]q_A(n)+\Delta_b,
\end{equation}
where $\Delta_b$ is also a constant depending on $R_{AB},R_{BA},\Sigma_{\min}, \Sigma_{AB}$ and $\Sigma_{BA}$.

Finally consider the case in \textsf{Step 3} of Algorithm \ref{Alg:OMLNC}. We have $q_A(n+1)=\check{\delta}_A(n)$ and $q_B(n+1)=q_B(n)-1+\check{\delta}_B(n)$, where $\check{\delta}_A(n)$ is a Poisson random variable with parameter $\frac{R_{AB}}{\Sigma_{BA}}$ and $\check{\delta}_B(n)$ is also Poisson with parameter $\frac{R_{BA}}{\Sigma_{BA}}$. So it follows that
\begin{equation}\label{Eqn:ExpeLyapFuncOPLNC02}
 \mathbb{E}\left[V(n+1)|\mathbf{Q}(n)\right] = V(n)+2\left[\frac{R_{AB}}{\Sigma_{BA}}+\frac{\Sigma_{\min}}{\Sigma_{BA}-\Sigma_{\min}}\left(\frac{R_{BA}}{\Sigma_{BA}}-1\right)\right]q_B(n)+\Delta_c,
\end{equation}
where $\Delta_c$ is also a constant depending on $R_{AB},R_{BA},\Sigma_{\min}, \Sigma_{AB}$ and $\Sigma_{BA}$. Since $(R_{AB},R_{BA})$ is within the quadrilateral formed by the origin, $(\Sigma_{AB},0)$, $(\Sigma_{\min},\Sigma_{\min})$ and $(0,\Sigma_{BA})$, the following inequalities are obvious:
\begin{eqnarray*}
  \frac{\Sigma_{\min}}{\Sigma_{AB}-\Sigma_{\min}}\left(\frac{R_{AB}}{\Sigma_{AB}}-1\right)+\frac{R_{BA}}{\Sigma_{AB}} < 0 \quad\text{and}\quad
  \frac{\Sigma_{\min}}{\Sigma_{BA}-\Sigma_{\min}}\left(\frac{R_{BA}}{\Sigma_{BA}}-1\right)+\frac{R_{AB}}{\Sigma_{BA}} < 0.
\end{eqnarray*}
Therefore, we have $\mathbb{E}[V(n+1)|\mathbf{Q}(n)]\leq V(n)-1$ when $q_A(n)$ or $q_B(n)$ is sufficiently large. According to the Foster-Lyapunov criterion \cite{SPMRLT93}, $\mathbf{Q}(n)$ is stable.
\end{proof}

Next, we consider the opportunistic packet scheduling algorithm for PLNC. We want to show that any rate pair in Theorem \ref{Thm:AchieRateRegiPLNC} is stabilizable. Consider a pair $(Q_A,Q_B)\in \mathbb{N}^2_+$ equal to $(\max Q_f,\max Q_b),$ where $Q_f,Q_b\in \mathbb{N}_+$ are subject to $Q_f\leq \lambda_1\Sigma_{AB}/\ell,\, Q_b\leq \lambda_2\Sigma_{BA}/\ell,\, \text{and }Q_b\lambda_1\Sigma_{ABA}=Q_f\lambda_2\Sigma_{BAB}$
\footnote{Theoretically, we should choose a positive integer pair $(Q_f,Q_b)$ large enough such that $(Q_f\ell/\lambda_1,Q_b\ell/\lambda_2)=(\Sigma_{ABA},\Sigma_{BAB})$. However, such a pair may not exist. So it is acceptable to select a pair with appropriate large values of $Q_f$ and $Q_b$ such that $(Q_f\ell/\lambda_1,Q_b\ell/\lambda_2)\approx(\Sigma_{ABA},\Sigma_{BAB})$.}. Let $(\tilde{\Sigma}_{ABA},\tilde{\Sigma}_{BAB})\in\mathcal{R}_{\texttt{PLNC} }$ be the closest point to $(\Sigma_{ABA},\Sigma_{BAB})$ with constraint $Q_A\lambda_2\tilde{\Sigma}_{BAB}\approx Q_B\lambda_1\tilde{\Sigma}_{ABA}$. Also, choose $Q^*\in \mathbb{N}_+$ large enough such that the following inequalities are satisfied:
\sublabon{equation}
\begin{eqnarray}\label{Eqn:ConsOPLNC}
  \frac{\tilde{\Sigma}_{ABA}}{\Sigma_{BA}-\tilde{\Sigma}_{BAB}}\left(\frac{R_{BA}}{\Sigma_{BA}}-1\right)+\frac{R_{AB}}{\Sigma_{BA}}+\frac{Q_A}{Q^*} &<& 0, \label{Eqn:ConsOPLNC01}\\
  \frac{\tilde{\Sigma}_{BAB}}{\Sigma_{AB}-\tilde{\Sigma}_{ABA}}\left(\frac{R_{AB}}{\Sigma_{AB}}-1\right)+\frac{R_{BA}}{\Sigma_{AB}}+\frac{Q_B}{Q^*} &<& 0,\label{Eqn:ConsOPLNC02}\\
  \max\{Q_A,Q_B\}-Q^* &<& 0\label{Eqn:ConsOPLNC03}.
\end{eqnarray}
\sublaboff{equation}
Such $Q^*$ must exist because the region confined by \eqref{Eqn:ConsOPLNC01} and \eqref{Eqn:ConsOPLNC02} in the first quadrant is enclosed by the achievable rate region specified in Theorem \ref{Thm:AchieRateRegiPLNC}. We propose the following algorithm that \emph{opportunistically} schedules packet transmissions over the relay node by PLNC.
\begin{algorithm}\label{Alg:OPLNC} \textbf{\emph{Opportunistic PLNC and Scheduling}}
\vspace{0.01in}\centering\small
\begin{tabular}{|lp{4.0in}|}\hline
\textsf{Step 1} ($q_A\geq Q_A$, $q_B\geq Q_B$): & $Q_A$ packets from node A and $Q_B$ packets from node B are sent over Relay D by PLNC.\\
\textsf{Step 2} ($q_A<Q_A$, $q_B>Q_B$): &  $\min\{q_B,Q^*\}$ packets from node B are sent over Relay D by TDMH.\\
\textsf{Step 3} ($q_A>Q_A$, $q_B<Q_B$): & $\min\{q_A,Q^*\}$ packets from node A are sent over Relay D by TDMH.\\
\textsf{Step 4} ($q_A<Q_A$, $q_B<Q_B$): & $q_A$ packets from node A and $q_B$ packets from node B are sent over Relay D by PLNC.\\ \hline
\end{tabular}
\end{algorithm}
\vspace{0.1in}

It will be also discussed in Section \ref{SubSec:SumRateOpti} that PLNC achieves its maximum network coding gain when the traffic pattern over a relay node is $\mu=\frac{C_{DA}}{C_{DB}}$. So when the two buffers respectively have at least $Q_A$ and $Q_B$ packets, Algorithm \ref{Alg:OPLNC} schedules the two-way traffic such that it has a traffic pattern $\mu \approx \frac{C_{DA}}{C_{DB}}$. If one of the two buffers does not have enough packets to achieve $\mu \approx \frac{C_{DA}}{C_{DB}}$, then one-way traffic is scheduled. However, if both of the buffers do not have enough packets then PLNC is adopted because it has a better throughput than TDMH for any traffic pattern as shown in Fig. \ref{Fig:AchiRateRegi}.

\begin{theorem}\label{Thm:StabOPLNC}
\emph{Algorithm \ref{Alg:OPLNC} stabilizes the two-way relaying system for Poisson arrivals with the (bit-arrival) rate pair $(R_{AB},R_{BA})$ within the region constructed by \eqref{Eqn:ConsOPLNC01}-\eqref{Eqn:ConsOPLNC03} where $R_{AB}=\acute{R}_{AB}\ell$ and $R_{BA}=\acute{R}_{BA}\ell$.}
\end{theorem}
\begin{proof}
The proof is omitted here since it is similar to the proof of Theorem \ref{Thm:StabOMLNC}.
\end{proof}

\begin{remark}
As shown in Theorem \ref{Thm:StabOPLNC}, Algorithm \ref{Alg:OPLNC} is able to stabilize the queues at the source nodes with Poisson arrivals whose rate pair is within the rate region constructed by \eqref{Eqn:ConsOPLNC01}-\eqref{Eqn:ConsOPLNC02} in the first quadrant, which is enclosed by $\mathcal{R}_{\texttt{PLNC}}$ in Theorem \ref{Thm:AchieRateRegiPLNC} for all $Q^*$ satisfied with \eqref{Eqn:ConsOPLNC03}. Therefore, any rate pair within $\mathcal{R}_{\texttt{PLNC}}$ in Theorem \ref{Thm:AchieRateRegiPLNC} is stabilizable since $Q^*$ can be chosen sufficiently large such that the region constructed by \eqref{Eqn:ConsOPLNC01}-\eqref{Eqn:ConsOPLNC02} approaches to $\mathcal{R}_{\texttt{PLNC}}$ as closely as possible. However, this might come with some queueing delay.
\end{remark}

\section{End-To-End Sum-Rate Optimization and Diversity-Multiplexing Tradeoff}\label{Sec:PerfAnal}
In this section, we characterize the end-to-end sum rates of the three transmission protocols with or without a traffic pattern constraint. Here we only consider the sum-rate optimization problem for a single relay since it is straightforward to extend to the multiple relay case. Next, we characterize the DMTs of the three transmission protocols with multiple relay nodes.

\subsection{End-to-End Sum-Rate Optimization}\label{SubSec:SumRateOpti}
We first consider the sum-rate optimization problems without any constraint on traffic pattern for TDMH, MLNC and PLNC. Then we consider how traffic pattern influences the sum rates. In addition, in order to compare the throughput performance we introduce the notion of \emph{network coding gain}, and characterize under what kind of traffic pattern the network coding gain is maximized. Considering Theorems \ref{Thm:AchieRateRegiTDMH}-\ref{Thm:AchieRateRegiPLNC} and Fig. \ref{Fig:AchiRateRegi}, then the following result is immediate for sum rate without traffic pattern constraints.
\begin{corollary}\label{Cor:MaxSumRate}
\emph{(a) The maximum sum rate achieved by TDMH is $\max\{\Sigma_{AB},\Sigma_{BA}\}$. \\(b) The maximum sum rate achieved by MLNC is
$\max\left\{\Sigma_{AB},\Sigma_{BB},2\Sigma_{ABB}\right\}$ if $C_{DB}<C_{DA}$, and if $C_{DB}\geq C_{DA}$, then it is $\max\left\{\Sigma_{BA},\Sigma_{AA},2\Sigma_{ABB}\right\}$. (c) The maximum sum rate achieved by PLNC is $\max\{\Sigma_{AB},\Sigma_{BA},\Sigma_{ABA}+\Sigma_{BAB}\}$.}
\end{corollary}
\begin{proof}
Since the end-to-end sum rate is $R_{AB}+R_{BA}$ and all constraints in \eqref{Eqn:AchiRateRegiTDMH}, \eqref{Eqn:AchiRateRegiMLNC} and \eqref{Eqn:AchiRateRegiPLNC} are linear, the maximum sum rate must be achieved at a corner point of the achievable rate region according to linear programming theory.
\end{proof}

For the sum rate optimization problem with a traffic pattern, we have the following.
\begin{corollary}\label{Cor:MaxSumRateMu}
\emph{If $\mu$ is fixed, the maximum sum rates achieved by TDMH, MLNC and PLNC are:}
\begin{eqnarray}
R^*_{\texttt{TDMH} } &=& (1+\mu)/\left(\Sigma^{-1}_{AB}+\mu \Sigma^{-1}_{BA}\right),\label{Eqn:OptiSumRateTDMH}\\
R^*_{\texttt{MLNC} } &=& (1+\mu)/\left(C^{-1}_{AD}+\mu C^{-1}_{BD}+\max\{1,\mu\}C^{-1}_{\min}\right),\label{Eqn:OptiSumRateMLNC}\\
R^*_{\texttt{PLNC} } &=& (1+\mu)/\left(C^{-1}_{AD}+\mu C^{-1}_{BD}+\max\{\mu C^{-1}_{DA},C^{-1}_{DB}\}\right)\label{Eqn:OptiSumRatePLNC}
\end{eqnarray}
\end{corollary}
\begin{proof}
For TDMH, its achievable rate region is enclosed in the first quadrant by line $\frac{R_{AB}}{\Sigma_{AB}}+\frac{R_{BA}}{\Sigma_{BA}}=1$ according to Theorem \ref{Thm:AchieRateRegiTDMH}. So replacing $R_{BA}$ with $\mu R_{AB}$ in the line equation and solving for $R_{AB}$, \eqref{Eqn:OptiSumRateTDMH} is obtained by $(1+\mu)R_{AB}$. For MLNC, we first consider the case when $C_{DB}<C_{DA}$. According to Theorem \ref{Thm:AchieRateRegiMLNC}, the boundary line of the achievable rate region is $\frac{R_{AB}}{\Sigma_{AB}}+\frac{R_{BA}}{C_{BD}}=1$ when $\mu<1$. Replacing $R_{BA}$ with $\mu R_{AB}$ and solving for $R_{AB}$, and then we can get the desired result for $\mu<1$. The boundary line is $\frac{R_{AB}}{C_{AD}}+\frac{R_{BA}}{\Sigma_{BB}}=1$ when $\mu \geq 1$, and hence the maximum sum rate can be found by replacing $R_{BA}$ with $\mu R_{AB}$. The case $C_{DB}\geq C_{DA}$ can be solved by the same fashion. For PLNC, its proof is omitted here since it is similar to the proof of MLNC.
\end{proof}

Given the results of the sum-rate optimization, the throughput of the three transmission protocols can be computed. We introduce the network coding (throughput) gain $\rho_{\texttt{AB}}$ (in dB scale) as
\begin{equation}\label{Defn:NetwCodiGain}
   \rho_{\texttt{AB}} \defn 10\log_{10}\frac{R^*_{\texttt{A}}}{R^*_{\texttt{B}}}\,\,(\text{dB}),
\end{equation}
where $R^*_{\texttt{A}}$ and $R^*_{\texttt{B}}$ are the maximum sum rates for protocol \texttt{A} and protocol \texttt{B}, respectively. Here if we consider TDMH as the baseline to be compared, the network coding gains of MLNC and PLNC are $\rho_{\texttt{MT}} = 10\log_{10}(R^*_{\texttt{MLNC}}/R^*_{\texttt{TDMH}})$ and $\rho_{\texttt{PT}} = 10\log_{10}(R^*_{\texttt{PLNC}}/R^*_{\texttt{TDMH}})$, respectively. Note that it is easy to check $\rho_{\texttt{PT}}>0$ and $\rho_{\texttt{PM}}\defn\rho_{\texttt{PT}}-\rho_{\texttt{MT}}>0$, which show PLNC is always superior to TDMH and MLNC. It is also easy to verify that $\rho_{\texttt{MT}}$ achieves its maximum at $\mu=1$ and $\rho_{\texttt{PT}}$ achieves its maximum at $\mu=\frac{C_{DA}}{C_{DB}}$ by Calculus. Recall that Algorithms \ref{Alg:OMLNC} and \ref{Alg:OPLNC} schedule packet transmissions according to the packet arrival processes at the source nodes. Their main idea is to schedule transmissions around such optimal points.

\subsection{Diversity-Multiplexing Tradeoffs of Two-Way Transmission Protocols}\label{SubSec:DMT}
In this subsection we investigate the DMTs of the three transmission protocols over multiple relays. Diversity gain $d$ and multiplexing gain $m$ in \cite{LZDNCT03} are redefined here in our notation as follows
$$d\defn -\lim_{\gamma\rightarrow\infty} \frac{\log \epsilon(\gamma)}{\log \gamma}\quad \text{and}\quad m \defn \lim_{\gamma\rightarrow\infty}\frac{R(\gamma)}{\log\gamma},$$
where $\epsilon$ is the outage probability of \emph{information exchange in two-way relaying}, $R$ is the \emph{end-to-end} \emph{transmission rate between source nodes}, and $\gamma$ is the signal-to-noise (SNR) ratio without fading. Note that $\epsilon$ and $R$ are not defined by the traditional fashion of \emph{point-to-point} transmission. They are defined by an end-to-end fashion because the three transmission protocols are multihop-based protocols, and the system we study here is aimed at the main point of information exchange so that it would be fairer for each transmission protocol to declare an outage in the system when either one source or both source nodes cannot receive the packets they desire. Thus, the outage probability of transmission protocol $\texttt{A}$ in the two-way relaying system is defined as
\begin{equation}\label{Eqn:DefOutProb2WaySys}
\epsilon_{\texttt{A}} \defn \mathbb{P}\left[\mathcal{E}_{\texttt{A},f}\bigcup\mathcal{E}_{\texttt{A},b}\right]
\end{equation}
where $\mathcal{E}_{\texttt{A},f}\defn \{\lambda_fI_{\texttt{A},f}<R_{AB}\}$ and $\mathcal{E}_{\texttt{A},b}\defn\{\lambda_bI_{\texttt{A},b}<R_{BA}\}$ are the outage events of forward and backward transmission, $I_{\texttt{A},f}$ ($I_{\texttt{A},b}$) is the mutual information of forward (backward) transmission and $\{\lambda_f,\lambda_b: \lambda_f,\lambda_b\in[0,1], \lambda_f+\lambda_b =1\}$ are time-allocation parameters for forward and backward transmission, respectively.

We first look at the DMT problem of TDMH over multiple relay nodes. Although a similar problem for one-way transmission has been investigated in \cite{JNLGWW03}\cite{ABAKDPRAL06}, our two-way results are more general as we can see in the following proposition. Considering Gaussian input distribution and the case of relay collaboration, then the forward and backward mutual information of the two-way relaying system in Fig. \ref{Fig:TwoWayRelaySys}(b) can be shown as
\begin{eqnarray}
    I_{\texttt{TDMH} ,f} = I_{\texttt{TDMH} ,b}=\frac{1}{2}\min\left\{I_{1}, I_{2}\right\}, \label{Eqn:CoopMutuInfoTDMHf}
\end{eqnarray}
where $I_{1}\defn \log\left(1+\gamma \sum_{D\in \mathcal{D}_{AB}} |h_{DA}|^2\right)$\footnote{To facilitate the DMT analysis here,
    all nodes are assumed to have the same transmit power. In fact, the DMT results are nothing to do with the transmit powers.} and $I_{2}\defn \log\left(1+\gamma \sum_{D\in \mathcal{D}_{AB}} |h_{DB}|^2\right)$ because the forward/backward transmission first virtually passes through a SIMO channel with receive MRC and then through a MISO channel with transmit MRC. Note that coefficient $\frac{1}{2}$ means the forward/backward data stream needs 2 time slots. Let $\mathcal{D}_{AB}\neq \emptyset$. Then we have the following result.
\begin{proposition}\label{Pro:DiveMutiTrofTDMH}
\emph{If all relay nodes in $\mathcal{D}_{AB}$ collaborate, then TDMH achieves the following DMT:
\begin{equation}\label{Eqn:CoopDiveMultTrofTDMH}
d = |\mathcal{D}_{AB}|\left(1-\frac{2m}{\min\{(1+\mu)\lambda_f,(1+1/\mu)\lambda_b\}}\right),
\end{equation}
where $m\in\left(0,\min\{(1+\mu)\lambda_f,(1+1/\mu)\lambda_b\}/2\right)$. If there is no collaboration in $\mathcal{D}_{AB}$, then TDMH over $D_{\texttt{TDMH} }^*$ is able to achieve the DMT in \eqref{Eqn:CoopDiveMultTrofTDMH} as well, where $D_{\texttt{TDMH} }^*\in\mathcal{D}_{AB}$ denotes the optimal relay node for bidirectional transmission.}
\end{proposition}

The proof of Proposition \ref{Pro:DiveMutiTrofTDMH} can be found in Appendix \ref{App:ProofThmDiveMultiTrofTDMH}. The results in Proposition \ref{Pro:DiveMutiTrofTDMH} are actually similar to the DMT results of the decode-and-forward protocol in \cite{JNLGWW03,ABAKDPRAL06} if we set $\{\mu=\infty,\lambda_b=1\}$ or $\{\mu=0,\lambda_f=1\}$, \ie considering the one-way transmission case. Proposition \ref{Pro:DiveMutiTrofTDMH} is more general by including the time-allocation and traffic pattern influences in two-way relaying.

The DMTs of MLNC and PLNC basically can be derived by the same way used in the proof of Proposition \ref{Pro:DiveMutiTrofTDMH}. Let us first consider the MLNC and PLNC protocols with relay collaboration. Although the relay nodes can achieve receive MRC, they have a difficulty broadcasting to attain \emph{bidirectional transmit MRC} at both source nodes simultaneously. So if all relay nodes broadcast at the same time, then the mutual information of forward and backward transmissions in this case are
\begin{eqnarray}
&&I_{\texttt{MLNC} ,f} = \frac{2}{3}\min\left\{I_1,\min\left\{\tilde{I}_{1}, \tilde{I}_{2}\right\}\right\}\quad\, \text{and}\quad
I_{\texttt{MLNC} ,b} = \frac{2}{3}\min\left\{I_2,\min\left\{\tilde{I}_{1}, \tilde{I}_{2}\right\}\right\},\label{Eqn:NonCoopMutuInfoMLNC}\\
&&I_{\texttt{PLNC} ,f} = \frac{2}{3}\min\left\{I_{1},\tilde{I}_{2}\right\}\hspace{0.75in}\quad \text{and}\quad
I_{\texttt{PLNC} ,b} = \frac{2}{3}\min\left\{\tilde{I}_{1},I_{2}\right\}, \label{Eqn:NonCoopMutuInfoPLNC}
\end{eqnarray}
where $\tilde{I}_{1}\defn\log\left(1+\gamma|\sum_{D\in\mathcal{D}_{AB}} h_{DA}|^2\right)$, $\tilde{I}_{2}\defn \log\left(1+\gamma|\sum_{D\in\mathcal{D}_{AB}} h_{DB}|^2\right)$ and coefficient $\frac{2}{3}$ is due to two data streams sharing three time slots. Obviously, MLNC and PLNC cannot achieve the full diversity gain $|\mathcal{D}_{AB}|$. Since nodes A and B loose transmit diversity from relay nodes in the broadcast stage, it is an suboptimal strategy to let all relay nodes broadcast the same information at the same time. The better idea is to find an optimal broadcast relay for MLNC and PLNC by the following criteria from Lemma \ref{Lem:BCRateNC}:
\begin{eqnarray}
D_\texttt{MLNC}^* = \arg\max_{D\in\mathcal{D}_{AB}}\min\{C_{DA},C_{DB}\}\quad \text{and}\quad D_\texttt{PLNC}^* = \arg\max_{D\in\mathcal{D}_{AB}}C_{DA}+C_{DB}.\label{Eqn:OptiBroadcastRelayNC}
\end{eqnarray}
In this case, $\tilde{I}_{1}=\log(1+\gamma|h_{A D^*}|^2)$ and $\tilde{I}_{2}=\log(1+\gamma|h_{B D^*}|^2)$ so that $I_1\geq \tilde{I}_1$,  $I_2 \geq \tilde{I}_2$. Thus \eqref{Eqn:NonCoopMutuInfoMLNC} reduces to $I_{\texttt{MLNC} ,f}=I_{\texttt{MLNC} ,b}=\frac{2}{3}\min\left\{\tilde{I}_{1}, \tilde{I}_{2}\right\}$ while \eqref{Eqn:NonCoopMutuInfoPLNC} remains unchanged. On the other hand, if two-way transmission is only over an optimally selected relay node, then \eqref{Eqn:NonCoopMutuInfoMLNC} and \eqref{Eqn:NonCoopMutuInfoPLNC} become
\begin{eqnarray}
I_{\texttt{MLNC} ,f} &=& I_{\texttt{MLNC} ,b} = \frac{2}{3}\min\left\{\log(1+\gamma|h_{AD^*_{\texttt{MLNC}}}|),\log(1+\gamma|h_{BD^*_{\texttt{MLNC}}}|)\right\}\label{Eqn:OptiRelMutuInfoMLNC},\\
I_{\texttt{PLNC} ,f} &=& I_{\texttt{PLNC} ,b} = \frac{2}{3}\min\left\{\log(1+\gamma|h_{AD^*_{\texttt{PLNC}}}|),\log(1+\gamma|h_{BD^*_{\texttt{PLNC}}}|)\right\}\label{Eqn:OptiRelMutuInfoPLNC},
\end{eqnarray}
where $D_{\texttt{MLNC}}^*$ and $D_{\texttt{PLNC}}^*$ are the optimal relay nodes determined by
\begin{eqnarray}\label{Eqn:OptiRelayNC}
  D^*_{\texttt{MLNC}}= \arg\max_{D\in\mathcal{D}_{AB}}\Sigma_{ABB}\quad \text{and }\quad D^*_{\texttt{PLNC}}=\arg\max_{D\in\mathcal{D}_{AB}}\Sigma_{ABA}+\Sigma_{BAB},
\end{eqnarray}
where $2\Sigma_{ABB}$ is the sum rate within $\mathcal{R}_{\texttt{MLNC} }$ and $\Sigma_{ABA}+\Sigma_{BAB}$ is the sum rate within $\mathcal{R}_{\texttt{PLNC} }$ at relay node $D\in\mathcal{D}_{AB}$. The DMTs for MLNC and PLNC with an optimally selected relay can thus be derived by using \eqref{Eqn:OptiRelMutuInfoMLNC} and \eqref{Eqn:OptiRelMutuInfoPLNC}. We summarize the DMTs of MLNC and PLNC for the scenarios with or without relay collaboration in Proposition \ref{Pro:CoopDiveMutiTrafNC}, and its proof is given in Appendix \ref{App:ProofThmCoopDiveMutiTrafNC}.

\begin{proposition}\label{Pro:CoopDiveMutiTrafNC}
\emph{If all relay nodes in $\mathcal{D}_{AB}$ collaborate to receive and broadcast at the same time, then the following DMT is achieved by MLNC and PLNC:
\begin{equation}\label{Eqn:CoopDiveMultTrofNC01}
d=1-\frac{3m}{2\min\{(1+\mu)\lambda_f,(1+1/\mu)\lambda_b\}}.
\end{equation}
where $m\in\left(0,\frac{2}{3}\min\{(1+\mu)\lambda_f,(1+1/\mu)\lambda_b\}\right)$. If all relay nodes collaborate to receive and an optimal relay node is selected by \eqref{Eqn:OptiBroadcastRelayNC} to broadcast, MLNC and PLNC achieve the following DMT:
\begin{equation}\label{Eqn:CoopDiveMultTrofNC02}
d=|\mathcal{D}_{AB}|\left(1-\frac{3m}{2\min\{(1+\mu)\lambda_f,(1+1/\mu)\lambda_b\}}\right).
\end{equation}
Furthermore, if an optimal relay node is selected by \eqref{Eqn:OptiRelayNC} to receive and broadcast for MLNC and PLNC, then the DMT in \eqref{Eqn:CoopDiveMultTrofNC02} is also achieved.}
\end{proposition}

The results in Proposition \ref{Pro:CoopDiveMutiTrafNC} are reasonable since PLNC is no longer superior to MLNC while broadcasting in the high SNR regime. In addition, due to relay selection diversity, using an optimally selected relay to broadcast (or to receive and then broadcast) is able to achieve the full diversity $|\mathcal{D}|_{AB}$. The results in Propositions \ref{Pro:DiveMutiTrofTDMH} and \ref{Pro:CoopDiveMutiTrafNC} have been presented in Fig. \ref{Fig:DMTradeoff} if we set $\mu=1$ for the three transmission protocols. First look at the results of solid lines in Fig. \ref{Fig:DMTradeoff}. They are obtained based on the optimal time allocation for MLNC and PLNC with two-way relaying, \ie $\lambda_f=\lambda_b=0.5$. MLNC and PLNC always have a better DMT than TDMH in any transmission case. However, the DMTs of MLNC and PLNC may not be better than TDMH if time allocation between forward and backward traffic is suboptimal. For example, suppose the four time allocation parameters of TDMH are $\lambda_1=\lambda_4=0.01$ and $\lambda_2=\lambda_3=0.49$ so that $\lambda_f=\lambda_b=\frac{0.01+0.49}{1}=0.5$. Thus \eqref{Eqn:CoopDiveMultTrofTDMH} becomes $d=|\mathcal{D}_{AB}|(1-2m)$ in this case. Similarly, if the three time allocation parameters for MLNC and PLNC are $\lambda_1=0.01$, $\lambda_2=0.49$ and thus $\lambda_3=0.5$ then MLNC and PLNC have $\lambda_f=\frac{\lambda_1+\lambda_3}{\lambda_1+\lambda_2+2\lambda_3}=0.34$ and $\lambda_b=1-\lambda_f=0.66$ in this case. So \eqref{Eqn:CoopDiveMultTrofNC02} becomes $d=|\mathcal{D}_{AB}|(1-2.2m)$ and TDMH thus has a better DMT than MLNC and PLNC, as the results of dashed line shown in Fig. \ref{Fig:DMTradeoff}. For two-way transmission, an ideal transmission protocol is that it can achieve full diversity gain $|\mathcal{D}_{AB}|$ and multiplexing gain one. The ideal DMT line in Fig. \ref{Fig:DMTradeoff} can be \emph{asymptotically} approached if there exists a transmission protocol which is able to support $N$ source nodes to exchange their packets in $N+1$ time slots even when $N$ is very large.

\section{Simulation Results}\label{Sec:SimuResults}
>From the analysis in Section \ref{Sec:PerfAnal}, we know the relationship between the sum rate and the traffic pattern parameter and the DMTs for the three transmission protocols. In this section, we use simulations to illustrate how the sum rate is affected by the traffic pattern in a two-way relaying system, and then show how outage probabilities behave as a function of SNR. We assume that all nodes have the same transmit power 18 dBm. The channel between any two nodes has path loss exponent 3.5 and is reciprocal with flat Rayleigh fading. The distance between source nodes A and B is 50m, and the relay nodes are randomly dropped on a 10m vertical line whose center is located at the middle point between nodes A and B.

We first look at the single relay case and consider the relay node is positioned at the middle point of the two source nodes. The simulation result of the network coding gains for this case is presented in Fig. \ref{Fig:SimuNCgain}. From the figure, we can observe that the maximum coding gains happen at $\mu\approx 1$ for all transmission protocols, and they are seriously impacted when the traffic is very asymmetric. Note that the maximum of $\rho_{\texttt{PT}}$ indeed happens at $\mu\approx\frac{C_{DA}}{C_{CB}}$ because $\frac{C_{DA}}{C_{CB}}$ is close to unity in our simulation setup. As expected, PLNC has not only the best network coding gain among all the protocols but also the robustness against asymmetric traffic pattern. The network coding gain of the opportunistic MLNC protocol, $\rho_{\texttt{OMT}}$ is also much better than what achieved by MLNC. In Fig \ref{Fig:SimuNCgain}, MLNC does not achieve a larger rate region than TDMH when $\mu\gg 1$ or $\mu\ll 1$. For example, when $\mu \notin (0.2,6.5)$ a positive network coding gain for MLNC does not exist, \ie $\rho_{\texttt{MT}}\leq 0$.

Now consider how the DMTs change with different transmission protocols over multiple relay nodes. In all the following simulations, we set the multiplexing gain $m=1/4$, $\mu=1$ and $\lambda_f=\lambda_b=1/2$. The outage probability of each protocol is calculated by $\mathbb{P}[\text{mutual information}<\frac{1}{8}\log\gamma]$ with $10^8$ channel realizations. First consider all relay nodes collaborate to receive and transmit (or broadcast) at the same time. The results shown in Fig. \ref{Fig:CoopOutageProb01} verify the DMTs in Proposition \ref{Pro:DiveMutiTrofTDMH} and \eqref{Eqn:CoopDiveMultTrofNC01} in Proposition \ref{Pro:CoopDiveMutiTrafNC}, \ie $d=\frac{1}{2}|\mathcal{D}_{AB}|$ for TDMH and $d=\frac{5}{8}$ for MLNC and PLNC. For example, the outage probability curve for TDHM with 3 relay nodes has a diversity gain about 1.5 and all the curves for MLNC and PLNC have the same diversity gain abut 0.6 close to $\frac{5}{8}$. Fig. \ref{Fig:NonCoopOutageProb} presents the results of selecting an optimal relay node to receive and broadcast. In the figure we can verify MLNC and PLNC have $d=\frac{5}{8}|\mathcal{D}_{AB}|$ which is better than $d=\frac{1}{2}|\mathcal{D}_{AB}|$ achieved by TDMH. The diversity gains of MLNC and PLNC are obviously greater than that of TDMH for any number of relay nodes for selection.

\section{Conclusion}\label{Sec:Conclusion}
The fundamental limits of information exchange over a two-way relaying channel with or without wireless network coding was investigated. We first characterized the achievable rate regions for the TDMH, MLNC and PLNC protocols, and found that MLNC does not always achieve a larger rate region than TDMH. An opportunistic time-sharing protocol between TDMH and MLNC is able to achieve $\mathbf{CovH}(\mathcal{R}_{\texttt{TDMH}},\mathcal{R}_{\texttt{MLNC}})$. The rate region achieved by PLNC is always larger than those achieved by TDMH and MLNC since PLNC can achieve the individual broadcast channel capacity. We then proposed two opportunistic packet scheduling algorithms for MLNC and PLNC that can stabilize the two-way relaying system for Poisson arrivals. The transmission performance of TDMH, MLNC and PLNC is also investigated in terms of the sum rates and the DMT. The maximum sum rates of the three protocols with or without a traffic pattern constraint were found. We showed that the three transmission protocols, by using the optimally selected relay node, can achieve their corresponding best DMTs achieved by relay collaboration. Also, we clarified that MLNC and PLNC may not always achieve a better DMT than TDMH due to suboptimal time allocation.

\appendices
\section{Proof of Proposition \ref{Pro:DiveMutiTrofTDMH}}\label{App:ProofThmDiveMultiTrofTDMH}
\begin{proof}
Consider the two-way relaying system in Fig.\ref{Fig:TwoWayRelaySys}(a). An outage is experienced when either $W_A$ or $W_B$ cannot be decoded correctly at their destination nodes. Let $\mathcal{E}_A$ ($\mathcal{E}_B$) denote the event that the relay node nodes in $\mathcal{D}_{AB}$ cannot correctly decode $W_A$ ($W_B$) and $\mathcal{E}^{c}_A$ ($\mathcal{E}^c_B$) denote the complement of $\mathcal{E}_A$ $(\mathcal{E}_B)$. According to \eqref{Eqn:DefOutProb2WaySys} and using Boole's inequality, a two-way relaying system has the inequality of outage probability: $\epsilon_{\texttt{TDMH}}\leq \epsilon_{\texttt{TDMH},f}+\epsilon_{\texttt{TDMH},b}$, where $\epsilon_{\texttt{TDMH},f}$ and $\epsilon_{\texttt{TDMH},b}$ are the forward and backward outage probabilities, respectively. Since we know $\epsilon_{\texttt{TDMH},f} = \mathbb{P}\left[\mathcal{E}_{\texttt{TDMH},f}|\mathcal{E}_A\right]\mathbb{P}[\mathcal{E}_A]+
\mathbb{P}\left[\mathcal{E}_{\texttt{TDMH},f}|\mathcal{E}^c_A\right]\mathbb{P}[\mathcal{E}^c_A]$,
$\epsilon_{\texttt{TDMH},f}= \mathbb{P}[\mathcal{E}_A]+\mathbb{P}\left[\frac{\lambda_f}{2}I_{2}<R_{AB}\right]\mathbb{P}[\mathcal{E}^c_A]$
where $\mathbb{P}[\mathcal{E}_A]=\mathbb{P}\left[\frac{\lambda_f}{2}I_{1}<R_{AB}\right]$. \emph{ Note that in the following analysis, we use notation $\gamma\star x$ instead of $\gamma^x$ in order to clearly present the complex expression of exponent $x$}. Let $R_{AB}+R_{BA}=m\log\gamma$ so that $R_{AB}=\frac{m}{1+\mu}\log\gamma$. Therefore,
\begin{eqnarray*}
\epsilon_{\texttt{TDMH},f} &\leq& 2\cdot\mathbb{P}\left[\min\left\{\sum_{D\in\mathcal{D}_{AB}} |h_{DA}|^2,\sum_{D\in\mathcal{D}_{AB}} |h_{DB}|^2\right\} < \gamma\star d_f\right] \overset{(a)}{\dotleq}\gamma\star\left(|\mathcal{D}_{AB}|d_f\right),
\end{eqnarray*}
for large $\gamma$ and $m\in\left(0,\frac{1}{2}(1+\mu)\lambda_f\right)$, where $d_f\defn \left(\frac{2m}{\lambda_f(1+\mu)}-1\right)$ and $(a)$ follows from the notation definition and Lemma \ref{Lem:ProbExpoRVsApprox} in Appendix \ref{App:ResuAnalDiveMultTrof}. Similarly,
$\epsilon_{\texttt{TDMH},b}\dotleq \gamma\star\left(|\mathcal{D}_{AB}|d_b\right),$
for large $\gamma$ and $m\in\left(0,\frac{1}{2}(1+1/\mu)\lambda_b\right)$ where $d_b \defn \left(\frac{2m}{(1+1/\mu)\lambda_b}-1\right)$. Therefore, \eqref{Eqn:CoopDiveMultTrofTDMH} can be obtained for large $\gamma$ and $m\in(0,\frac{1}{2}\min\{(1+\mu)\lambda_f,(1+1/\mu)\lambda_b\})$. Next, consider there is no collaboration in $\mathcal{D}_{AB}$ and an optimal relay is selected to receive and broadcast. In this case \eqref{Eqn:CoopMutuInfoTDMHf} becomes
\begin{eqnarray}
I_{\texttt{TDMH} ,f} = I_{\texttt{TDMH} ,b} = \frac{1}{2}\min\left\{\log\left(1+\gamma|h_{AD_{\texttt{TDMH} }^*}|^2\right),  \log\left[1+\gamma|h_{D_{\texttt{TDMH} }^*B}|^2\right]\right\},\label{Eqn:NonCoopMutuInfoTDMHf}
\end{eqnarray}
where the optimal relay node $D_{\texttt{TDMH} }^*$ is selected according to the following criterion:
\begin{equation*}
    D_{\texttt{TDMH} }^* = \arg\min_{D\in\mathcal{D}_{AB}} (1/C_{DA}+1/C_{DB}) = \arg\max_{D\in\mathcal{D}_{AB}} \frac{|h_{DA}|^2|h_{DB}|^2}{|h_{DA}|^2+|h_{DB}|^2}.
\end{equation*}
That is to choose the relay node with $\max \{\Sigma_{AB}\}$ in $\mathcal{D}_{AB}$. Moreover by substituting \eqref{Eqn:NonCoopMutuInfoTDMHf} into $\epsilon_{\texttt{TDMH},f}$ and $R_{AB}=\frac{m}{1+\mu}\log\gamma$, we can obtain
\begin{eqnarray*}
\epsilon_{\texttt{TDMH},f} \leq 2\cdot\mathbb{P}\left[\frac{|h_{AD_{\texttt{TDMH} }^*}|^2|h_{BD_{\texttt{TDMH} }^*}|^2}{|h_{AD_{\texttt{TDMH} }^*}|^2+|h_{BD^*_{\texttt{TDMH}}}|^2}<\gamma\star d_f\right]\overset{(b)}{\dotleq} \gamma\star\left(|\mathcal{D}_{AB}|d_f\right),
\end{eqnarray*}
where $(b)$ follows from the fact that $D_{\texttt{TDMH}}^*\in\mathcal{D}_{AB}$ is optimal and Lemmas \ref{Lem:ProbExpoRVsApprox} and \ref{Lem:ProbOptiRandVectApprox} in Appendix \ref{App:ResuAnalDiveMultTrof}. Likewise, we can get a similar result for $\epsilon_{\texttt{TDMH},b}$ as shown in above. So \eqref{Eqn:CoopDiveMultTrofTDMH} can be concluded.
\end{proof}

\section{Proof of Proposition \ref{Pro:CoopDiveMutiTrafNC}}\label{App:ProofThmCoopDiveMutiTrafNC}
\begin{proof}
Using the same definitions of $\mathcal{E}_A$ and $\mathcal{E}_B$ in the proof of Proposition \ref{Pro:DiveMutiTrofTDMH}, the outage probability for MLNC and PLNC satisfies the inequality: $\epsilon_{\texttt{NC}} \leq \epsilon_{\texttt{NC},f}+\epsilon_{\texttt{NC},b}$, where ``\texttt{NC}'' means \texttt{MLNC}  or \texttt{PLNC}. Here we only prove the DMT of MLNC since the DMT of PLNC can be proved by following the same steps. According to the proof of Proposition \ref{Pro:DiveMutiTrofTDMH}, it is easy to show the following inequality for $\epsilon_{\texttt{MLNC},f}$:
\begin{eqnarray}
\epsilon_{\texttt{MLNC},f} \leq\mathbb{P}[\mathcal{E}_A]+\mathbb{P}\left[\frac{2}{3}\lambda_f \tilde{I}_{1}<R_{AB}\right]+\mathbb{P}\left[\frac{2}{3}\lambda_f \tilde{I}_{2}<R_{AB}\right],\label{Eqn:ErrProbMLNCf}
\end{eqnarray}
where $\mathbb{P}[\mathcal{E}_A]=\mathbb{P}\left[\frac{2}{3}\lambda_fI_{1}<R_{AB}\right]$. Let $R_{AB}=\frac{m}{1+\mu}\log\gamma$ and consider the first case that every relay node collaborates to receive and then broadcasts at the same time. Considering \eqref{Eqn:NonCoopMutuInfoMLNC} and using Lemma \ref{Lem:ProbExpoRVsApprox}, for large $\gamma$ and $m\in(0,2(1+\mu)\lambda_f/3)$, \eqref{Eqn:ErrProbMLNCf} becomes
\begin{eqnarray*}
\epsilon_{\texttt{MLNC},f}\leq \frac{1}{|\mathcal{D}_{AB}|!}\prod_{D\in \mathcal{D}_{AB}}\sigma_{AD}\left[\gamma\star\left(|\mathcal{D}_{AB}|\tilde{d}_f\right)\right]+(\sigma_{f_1}+\sigma_{f_2})\left(\gamma\star \tilde{d}_f\right) \dotleq \gamma\star \tilde{d}_f,
\end{eqnarray*}
where $\tilde{d}_f\defn \left[\frac{3m}{2(1+\mu)\lambda_f}-1\right]$, $1/\sigma_{f_1}$ and $1/\sigma_{f_2}$ are the variances of $|\sum_{D\in\mathcal{D}_{AB}}h_{AD}|^2$ and $|\sum_{D\in\mathcal{D}_{AB}}h_{BD}|^2$, respectively. Similarly, we can show that $\epsilon_{\texttt{MLNC},b}\dotleq \gamma\star \tilde{d}_b$ where $\tilde{d}_b \defn \left[\frac{3m}{2(1+1/\mu)\lambda_b}-1\right]$. Then DMT in \eqref{Eqn:CoopDiveMultTrofNC01} can be arrived. Now consider MLNC with optimal relay $D_{\texttt{MLNC}}^*$ selected by \eqref{Eqn:OptiBroadcastRelayNC} to broadcast. Then we have the following
\begin{eqnarray*}
\epsilon_{\texttt{MLNC},f} &=& \mathbb{P}\left[|h_{D_{\texttt{MLNC}}^*B}|^2<\gamma\star \tilde{d}_f\right]\overset{(a)}{=} \prod_{D\in\mathcal{D}_{AB}}\mathbb{P}\left[|h_{DB}|^2<\gamma\star \tilde{d}_f\right]\overset{(b)}{\dotleq} \gamma\star\left(|\mathcal{D}_{AB}|\tilde{d}_f\right),
\end{eqnarray*}
where $(a)$ follows from the fact that $D^*_{\texttt{MLNC}}$ is optimal and $\{h_{DB}\}$ are independent, and $(b)$ follows from  Lemma \ref{Lem:ProbExpoRVsApprox} in Appendix \ref{App:ResuAnalDiveMultTrof}. Similarly, we have $\epsilon_{\texttt{MLNC},b}\dotleq \gamma\star\left(|\mathcal{D}_{AB}|\tilde{d}_b\right).$
So \eqref{Eqn:CoopDiveMultTrofNC02} is obtained.

Next, we look at the DMT of MLNC when an optimal relay node $D^*_{\texttt{MLNC}}$ is selected to receive and broadcast. $D^*_{\texttt{MLNC}}$ is determined by \eqref{Eqn:OptiRelayNC} when $\gamma$ is large. Likewise, the first step is to calculate the inequality of the outage probability of the forward transmission with $R_{AB}=\frac{m}{1+\mu}\log\gamma$. According to \eqref{Eqn:OptiRelMutuInfoMLNC} and \eqref{Eqn:OptiRelMutuInfoPLNC}, the inequality of $\epsilon_{\texttt{MLNC},f}$ is obtained as follows
\begin{eqnarray}
&&\mathbb{P}\left[\frac{2}{3}\lambda_f\log\left(1+\gamma|h_{A D^*_{\texttt{MLNC}}}|^2\right)<R_{AB}\right]
\overset{(c)}{\leq}
\prod_{D\in\mathcal{D}_{AB}}\mathbb{P}\left[\frac{|h_{A D}|^2|h_{DB}|^2}{2|h_{DA}|^2+|h_{DB}|^2}<\gamma\star \tilde{d}_f\right],
\end{eqnarray}
where $(c)$ follows that $D^*_{\texttt{MLNC}}$ is optimal and channel gains are independent, and we thus have $\epsilon_{\texttt{MLNC},f}\dotleq \gamma\star\left(|\mathcal{D}_{AB}|\tilde{d}_f\right)$ due to Lemma \ref{Lem:ProbOptiRandVectApprox} in Appendix \ref{App:ResuAnalDiveMultTrof}. Similarly, we also can have $\epsilon_{\texttt{MLNC},b}\dotleq \gamma\star\left(|\mathcal{D}_{AB}|\tilde{d}_b\right)$. Accordingly, \eqref{Eqn:CoopDiveMultTrofNC02} can be concluded. This completes the proof.
\end{proof}

\section{Lemmas for DMT analysis}\label{App:ResuAnalDiveMultTrof}
We need the following definition before proceeding to prove the lemmas in this section.

\emph{\textbf{Definition}}: A function $g(\gamma):\mathbb{R}_{++}\rightarrow\mathbb{R}_{++}$ is said to be exponentially equal to $x$ (\ie $g(\gamma)\,\doteq\, \gamma^x$) if $\lim_{\gamma\rightarrow\infty}\frac{\log g(\gamma)}{\log\gamma}\,=\, x.$
Similar definition can be applied to other signs, such as $\leq$ or $\geq$.
\begin{lemma}\label{Lem:ProbExpoRVsApprox}
Let $\{X_k, k=1,\ldots,K\}$ be $K$ independent exponential random variables with respective parameter $\{\sigma_k, k=1,\ldots,K\}$ and $\theta(\gamma):\mathbb{R}_{++}\rightarrow \mathbb{R}_{++}$. If $\theta(\gamma)\rightarrow 0$ as $\gamma\rightarrow\infty$ and $\theta(\gamma)$ is exponentially equal to $\theta_{\infty}$, then we have the following inequality:
\begin{equation}
\mathbb{P}\left[\sum_{k=1}^{K}X_k < \theta(\gamma)\right] \dotleq \gamma^{K\theta_{\infty}}.
\end{equation}
\end{lemma}
\begin{proof}
Without loss of generality, assume that the random sequence $\{X_k, k=1,\cdots,K\}$ forms an order statistics, \ie $\{X_1\geq X_2\geq\cdots\geq X_K\}$. Thus, the event $\sum_{k=1}^K X_k \leq \theta(\gamma)$ is equivalent to the intersection event of $X_1\leq \theta(\gamma)$, $X_1+X_2\leq\theta(\gamma)$,..., $\sum_{k=1}^K X_k \leq \theta(\gamma)$ because $X_k\geq 0$, for all $k\in[1,\cdots,K]$. That is,
\begin{equation*}
\mathbb{P}\left[\sum_{k=1}^K X_k \leq \theta(\gamma)\right]=\mathbb{P}\left[\bigcap_{k=1}^{K}\left(\sum_{j=1}^{k}X_j\leq\theta(\gamma)\right)\right]\leq
\mathbb{P}\left[\bigcap_{k=1}^{K} \left(k \cdot X_k\leq\theta(\gamma)\right)\right] \overset{(a)}{=} \prod_{i=1}^{K}\mathbb{P}\left[X_k\leq\frac{\theta(\gamma)}{k}\right],
\end{equation*}
where $(a)$ follows from the independence between all random variables. Since all random variables are exponential, then we further have
\begin{equation}
\mathbb{P}\left[X_k\leq\frac{\theta(\gamma)}{k}\right]=1-\exp\left(-\sigma_k\frac{\theta(\gamma)}{k}\right)\overset{(b)}{\leq} \sigma_k\frac{\theta(\gamma)}{k},
\end{equation}
where $(b)$ follows from the fact that $X_k$ is an exponential random variable with parameter $\sigma_k$ and $e^{-y}\geq 1-y,\,\forall y\in\mathbb{R}_+$. By Definition and $\lim_{\gamma\rightarrow\infty}\log\theta(\gamma)/\log\gamma=\theta_{\infty}$, it follows that
$\mathbb{P}\left[\sum_{k=1}^K X_k \leq \theta(\gamma)\right] \leq \prod_{k=1}^{K} \frac{\sigma_k\theta(\gamma)}{k}
=[\theta(\gamma)]^K \frac{1}{K!}\prod_{k=1}^{K}\sigma_k \dotleq \gamma^{K\theta_{\infty}}$.
The proof is complete.
\end{proof}

\begin{lemma}\label{Lem:ProbOptiRandVectApprox}
Let $\mathcal{T}$ be a given countable finite set with cardinality $|\mathcal{T}|$ and $\mathcal{V}$ be a random vector set whose elements are $m$-tuples, independent and nonnegative, \ie $\mathcal{V}\defn\{\mathbf{V}_i, i\in\mathbb{N}_{+}: \mathbf{V}_i\in\mathbb{R}_+^m, \mathbf{V}_i\bot\mathbf{V}_j, i\neq j\}$. Suppose $\forall t\in\mathcal{T}$, $\mathbf{V}_t=(V_{t_1},V_{t_2},\ldots,V_{t_m})^{\top}\in\mathcal{V}$ is an exponential random vector with $m$ independent entries and $\gamma,\theta(\gamma)\in\mathbb{R}_{++}$. $\theta(\gamma)$, $\{\alpha_i(\gamma)\}$ and $\{\beta_i(\gamma)\}$ are exponentially equal to $\theta_{\infty}$, $\{\alpha_{i_{\infty}}\}$ and $\{\beta_{i_{\infty}}\}$, respectively. Let $f(\mathbf{V}_t)$ and $\tilde{f}(V_t)$ be respectively defined as follows:\\
\begin{equation*}
   f(\mathbf{V}_t) \defn \frac{\prod_{i=1}^m V_{t_i} [\sum_{i=1}^m \alpha_i(\gamma)V_{t_i}]}{\prod_{i=1}^m V_{t_i}+\sum_{i=1}^m \beta_i(\gamma)(V_{t_i})^m}\,\,\text{ and }\,\, \tilde{f}(\mathbf{V}_t) \defn \frac{\prod_{i=1}^m V_{t_i}}{\sum_{i=1}^m \beta_i(\gamma)(V_{t_i})^m}.
\end{equation*}
Suppose $t^* \defn \arg\max_{t\in\mathcal{T}} f(\mathbf{V}_t)$ and $\tilde{t}^* \defn \arg\max_{t\in\mathcal{T}} \tilde{f}(\mathbf{V}_t)$. If $\theta(\gamma)\rightarrow 0$ as $\gamma\rightarrow\infty$, then for sufficient large $\gamma$ we have
\begin{eqnarray}
\mathbb{P}\left[f(\mathbf{V}_{t^*})<\theta(\gamma)\right] &\dotleq& \gamma^{|\mathcal{T}|[\theta_{\infty}+m(\beta^+_{\max}-\alpha_{\max})]},\label{Eqn:DotLeqFunVt1}\\
\mathbb{P}\left[\tilde{f}(\mathbf{V}_{\tilde{t}^*})<\theta(\gamma)\right] &\dotleq& \gamma^{|\mathcal{T}|(\theta_{\infty}+m\beta^+_{\max})},\label{Eqn:DotLeqFunVt2}
\end{eqnarray}
where $\alpha_{\max} \defn \max_{i}\{\alpha_{i_{\infty}}\}$ and $\beta^+_{\max} \defn \max_{i}\{\beta_{i_{\infty}},0\}$.
\end{lemma}
\begin{proof}
Since  all random vectors in $\mathcal{V}$ are independent and $t^* = \arg\max_{t\in\mathcal{T}} f(\mathbf{V}_t)$, we have
\begin{equation}\label{Eqn:ProbRandVect01}
    \mathbb{P}\left[f(\mathbf{V}_{t^*})<\theta(\gamma)\right]=\mathbb{P}\left[f(\mathbf{V}_t)<\theta(\gamma), \forall t\in\mathcal{T}\right]=\prod_{t\in\mathcal{T}}\mathbb{P}\left[f(\mathbf{V}_t)<\theta(\gamma)\right].
\end{equation}
In addition, for any $t\in\mathcal{T}$ it is easy to show that
$$f(\mathbf{V}_t) \geq \frac{\sum_{i=1}^m\alpha_i(\gamma)}{1+\sum_{i=1}^m\beta_i(\gamma)}V_{t_{\max}}\left(\frac{V_{t_{\min}}}{V_{t_{\max}}}\right)^{m+1}=\phi_m(\gamma) V_{t_{\max}}\Psi_{t_m} ,$$
where $V_{t_{\min}} \defn \min\{\mathbf{V}_t\}$, $V_{t_{\max}} \defn \max\{\mathbf{V}_t\}$, $\phi_m(\gamma)\defn \frac{\sum_{i=1}^m\alpha_i(\gamma)}{1+\sum_{i=1}^m\beta_i(\gamma)}$ and $\Psi_{t_m} \defn \left(\frac{V_{t_{\min}}}{V_{t_{\max}}}\right)^{m+1}$. Therefore,
$\mathbb{P}\left[f(\mathbf{V}_t)<\theta(\gamma)\right] \leq \mathbb{P}\left[V_{t_{\max}}\Psi_{t_m} < \phi^{-1}_m(\gamma)\theta(\gamma)\right]\leq
\prod_{i=1}^m \mathbb{P}\left[V_{t_i} \Psi_{t_m} < \phi^{-1}_m(\gamma) \theta(\gamma)\right].$
Also,
\begin{eqnarray}\label{Eqn:ProbIneqV}
\mathbb{P}[V_{t_i} \Psi_{t_m} < \phi^{-1}_m(\gamma) \theta(\gamma)]&=&\int_{\mathbb{R}_{++}} \left[1-\exp\left(-\frac{\sigma_{t_i}\theta(\gamma)}{\phi_m(\gamma)\psi_{t_m}}\right)\right] f_{\Psi_{t_m}}(\psi_{t_m})\, d\psi_{t_m}\nonumber \\
&\overset{(c)}{\leq}& \int_{\mathbb{R}_{++}} \frac{\sigma_{t_i}\theta(\gamma)}{\phi_m(\gamma)\psi_{t_m}}f_{\Psi_{t_m}}(\psi_{t_m})\, d\psi_{t_m} = \frac{\sigma_{t_i}\theta(\gamma)}{\phi_m(\gamma)}\mathbb{E}\left[\frac{1}{\Psi_{t_m}}\right],
\end{eqnarray}
where $(c)$ follows from $e^{-x}\geq 1-x, \forall x\in\mathbb{R}_{+}$ and $\sigma_{t_i}$ is the parameter for $V_{t_i}$. For sufficiently large $\gamma$, it follows that
\begin{equation}
\mathbb{P}\left[f(\mathbf{V}_t)<\theta(\gamma)\right]\leq \left(\mathbb{E}\left[\frac{1}{\Psi_{t_m}}\right]\right)^{m}\prod^m_{i=1} \sigma_{t_i}\cdot \left(\phi_m(\gamma)\right)^{-m}\theta(\gamma).
\end{equation}
So for large $\gamma$ \eqref{Eqn:ProbRandVect01} can be rewritten as \eqref{Eqn:DotLeqFunVt1}. For any $t\in\mathcal{T}$, it is easy to show that
$$\tilde{f}(\mathbf{V}_t) \geq \frac{\prod_{i=1}^{m}V_{t_i}}{\prod_{i=1}^{m}V_{t_i}+\sum^m_{i=1}\beta_i(\gamma)(V_{t_i})^m}\defn \hat{f}(\mathbf{V}_t).$$
By considering $f(\mathbf{V}_t)$ with constant $\{\alpha_i(\gamma)\}$ (so $\alpha_{\max}=0$) and the result in \eqref{Eqn:DotLeqFunVt1},
we thus have
\begin{equation}
\mathbb{P}\left[\hat{f}(\mathbf{V}_t)<\theta(\gamma)\right] \dotleq \gamma^{(\theta_{\infty}+m\beta^+_{\max})},\,\,\,\text{and}\,\,\, \mathbb{P}\left[\tilde{f}(\mathbf{V}_{t^*}) < \theta(\gamma)\right] \leq \prod_{t\in\mathcal{T}}
\mathbb{P}\left[\hat{f}(\mathbf{V}_t)<\theta(\gamma)\right].\label{Eqn:IneqDotLeqFunVt2}
\end{equation}
Therefore, \eqref{Eqn:DotLeqFunVt2} can be arrived by \eqref{Eqn:IneqDotLeqFunVt2}. The proof is complete.
\end{proof}

\bibliographystyle{IEEEtran}
\bibliography{IEEEabrv,Ref_NCTWR}

\begin{thebibliography}{10}
\providecommand{\url}[1]{#1}
\csname url@samestyle\endcsname
\providecommand{\newblock}{\relax}
\providecommand{\bibinfo}[2]{#2}
\providecommand{\BIBentrySTDinterwordspacing}{\spaceskip=0pt\relax}
\providecommand{\BIBentryALTinterwordstretchfactor}{4}
\providecommand{\BIBentryALTinterwordspacing}{\spaceskip=\fontdimen2\font plus
\BIBentryALTinterwordstretchfactor\fontdimen3\font minus
  \fontdimen4\font\relax}
\providecommand{\BIBforeignlanguage}[2]{{%
\expandafter\ifx\csname l@#1\endcsname\relax
\typeout{** WARNING: IEEEtran.bst: No hyphenation pattern has been}%
\typeout{** loaded for the language `#1'. Using the pattern for}%
\typeout{** the default language instead.}%
\else
\language=\csname l@#1\endcsname
\fi
#2}}
\providecommand{\BIBdecl}{\relax}
\BIBdecl

\bibitem{RANCSYRLRWY00}
R.~Ahlswede, N.~Cai, S.-Y.~R. Li, and R.~W. Yeung, ``Network information
  flow,'' \emph{{IEEE} Trans. Inf. Theory}, vol.~46, no.~4, pp. 1204--1216,
  Jul. 2000.

\bibitem{CFES0701}
C.~Fragouli and E.~Soljanin, \emph{Network Coding Fundamentals}.\hskip 1em plus
  0.5em minus 0.4em\relax Hanover, MA, USA: Now Publisher Inc., 2007.

\bibitem{CFES0702}
------, \emph{Network Coding Applications}.\hskip 1em plus 0.5em minus
  0.4em\relax Hanover, MA, USA: Now Publisher Inc., 2007.

\bibitem{SKHRWHDKMMJC08}
S.~Katti, H.~Rahul, W.~Hu, D.~Katabi, M.~M\'{a}dard, and J.~Crowcroft, ``{XORs}
  in the air: Practical wireless network coding,'' \emph{{IEEE/ACM} Trans.
  Netw.}, vol.~16, no.~3, pp. 497--510, Jun. 2008.

\bibitem{DNTNBB99}
D.~Nguyen, T.~Nguyen, and B.~Bose, ``Wireless broadcasting using network
  coding,'' \emph{{IEEE} Trans. Veh. Technol.}, to appear.

\bibitem{PLNJKES06}
P.~Larsson, N.~Johansson, and K.-E. Sunell, ``Coded bi-direction relaying,'' in
  \emph{the Proc. of IEEE VTC.}, Spring 2006.

\bibitem{CHLFXSS08}
C.-H. Liu and F.~Xue, ``Network coding for two-way relaying: rate regions, sum
  rate and opportunistic scheduling,'' in \emph{the Proc.of IEEE International
  Conf. on Comm. (ICC)}, May 2008.

\bibitem{FXCHLSS07}
F.~Xue, C.-H. Liu, and S.~Sandhu, ``{MAC}-layer and {PHY}-layer network coding
  for two-way relaying: achievable rate regions and opportunistic scheduling,''
  in \emph{the Proc.of Allerton Conf. on Comm., Control and Computing}, Sep.
  2007.

\bibitem{PPHY0107}
P.~Popovski and H.~Yomo, ``Wireless network coding by amplify-and-forward for
  bi-directional traffic flows,'' \emph{{IEEE} Commun. Lett.}, vol.~11, no.~1,
  pp. 16--18, Jan. 2007.

\bibitem{RWYSYRLNCZZ06}
R.~W. Yeung, S.-Y.~R. Li, N.~Cai, and Z.~Zhang, \emph{Network Coding
  Theory}.\hskip 1em plus 0.5em minus 0.4em\relax Now Publisher Inc., 2006.

\bibitem{SYRLRWYNC03}
S.-Y.~R. Li, R.~W. Yeung, and N.~Cai, ``Linear network coding,'' \emph{{IEEE}
  Trans. Inf. Theory}, vol.~49, no.~2, pp. 371--381, Feb. 2003.

\bibitem{CHJH06}
C.~Hausl and J.~Hagenauer, ``Interative network coding and channel decoding for
  the two-way relay channel,'' in \emph{the Proc. of IEEE International
  conference on Comm. (ICC)}, Jun. 2006.

\bibitem{FXSS07}
F.~Xue and S.~Sandhu, ``{PHY}-layer network coding for broadcast channel with
  side information,'' in \emph{the Proc.of IEEE Information Theory Workshop},
  Sep. 2007.

\bibitem{CSTJOSS07}
C.~Schnurr, T.~J. Oechtering, and S.~Stanczak, ``On coding for the broadcast
  phase in the two-way relay channel,'' in \emph{the Proc. of Conf. on
  Information Sciences and Systems}, Mar. 2007.

\bibitem{GKSS07}
G.~Kramer and S.~Shamai(Shitz), ``Capacity for classes of broadcast channels
  with receiver side information,'' in \emph{the Proc.of IEEE Information
  Theory Workshop}, Sep. 2007.

\bibitem{SJKPMVT08}
S.~J. Kim, P.~Mitran, and V.~Tarokh, ``Performance bounds for bi-directional
  coded cooperation protocols,'' \emph{submitted}, 2008.

\bibitem{SKDKWHMMJC0905}
S.~Katti, D.~Katabi, W.~Hu, H.~Rahul, M.~M\'{a}dard, and J.~Crowcroft, ``The
  importance of being opportunistic: practical network coding for wireless
  environments,'' in \emph{the Proc. of Allerton Conf. on Comm., Control and
  Computing}, Sep. 2005.

\bibitem{SKSGDK07}
S.~Katti, S.~Gollakota, and D.~Katabi, ``Embracing wireless interference:
  analog network coding,'' in \emph{the Proc. of ACM SIGCOMM}, Aug. 2007.

\bibitem{PPHY0607}
P.~Popovski and H.~Yomo, ``Physical network coding in two-way wireless relay
  channel,'' in \emph{the Proc. of IEEE International Conf. on Comm. (ICC)},
  Jun. 2007.

\bibitem{MCAY08}
M.~Chen and A.~Yener, ``Multiuser two-way relaying: detection and interference
  management strategies,'' \emph{submitted to IEEE Trans. on wireless
  communications}, 2008.

\bibitem{SKIMAGDKMM07}
S.~Katti, I.~Mari\'{c}, A.~Goldsmith, D.~Katabi, and M.~M\'{e}dard, ``Joint
  relaying and network coding in wireless networks,'' in \emph{the Proc. of
  IEEE International Symposium on Information Theory}, Jun. 2007.

\bibitem{LXTEFJKDJC07}
L.~Xiao, T.~E. Fuja, J.~Kliewer, and D.~J. Costello, ``A network coding
  approach to cooperative diversity,'' \emph{{IEEE} Trans. Inf. Theory},
  vol.~53, no.~10, pp. 3714--3722, Oct. 2007.

\bibitem{YCSKJL06}
Y.~Chen, S.~Kishore, and J.~Li, ``Wireless diversity through network coding,''
  in \emph{the Proc. of IEEE Wireless Communication and Networking Conf.}, Apr.
  2006.

\bibitem{LZDNCT03}
L.~Zheng and D.~N.~C. Tse, ``Diversity and multiplexing: a fundamental tradeoff
  in multiple-antenna channles,'' \emph{{IEEE} Trans. Inf. Theory}, vol.~49,
  no.~5, pp. 1073--1096, May 2003.

\bibitem{DGAGHVP08}
D.~Gunduz, A.~Goldsmith, and H.~V. Poor, ``{MIMO} two-way relay channel:
  Diversity-multiplexing tradeoff analysis,'' in \emph{the Proc. of the
  Asilomar Conference on Signals, Systems and Computers}, Oct. 2008.

\bibitem{RVRWH08}
R.~Vaze and R.~W. Heath, ``On the capacity and diversity-multiplexing tradeoff
  of the two-way relay channel,'' Oct. 2008, submitted, available at
  http://arxiv.org/abs/0810.3900.

\bibitem{ABAKDPRAL06}
A.~Bletsas, A.~Khisti, D.~P. Reed, and A.~Lippman, ``A simple cooperative
  diversity method based on network path selection,'' \emph{{IEEE} J. Sel.
  Areas Commun.}, vol.~24, no.~3, pp. 659--672, Mar. 2006.

\bibitem{JNLGWW03}
J.~N. Laneman and G.~W. Wornell, ``Distributed space-time-coded protocols for
  exploiting cooperative diversity in wireless networks,'' \emph{{IEEE} Trans.
  Inf. Theory}, vol.~49, no.~10, pp. 2415--2425, Oct. 2003.

\bibitem{OOSS06}
O.~Oyman and S.~Sandhu, ``A shannon-theoretic perspective on fading multihop
  networks,'' in \emph{the Proc. of Conference on Information Sciences and
  Systems}, Mar. 2006.

\bibitem{MSJNLMHDJCTF06}
M.~Sikora, J.~N. Laneman, M.~Haenggi, D.~J. Costello, and T.~Fuja,
  ``Bandwidth-and power-efficient routing in linear wireless networks,''
  \emph{Joint Special Issue of IEEE Trans. on Information Theory and IEEE
  Trans. on Networking}, 2006.

\bibitem{SPMRLT93}
S.~P. Meyn and R.~L. Tweedie, \emph{Markov Chains and Stochastic
  Stability}.\hskip 1em plus 0.5em minus 0.4em\relax New York: Springer-Verlag,
  1993.

\end{thebibliography}
\newpage

\begin{figure}[h]
\centering
  \includegraphics[scale=0.8]{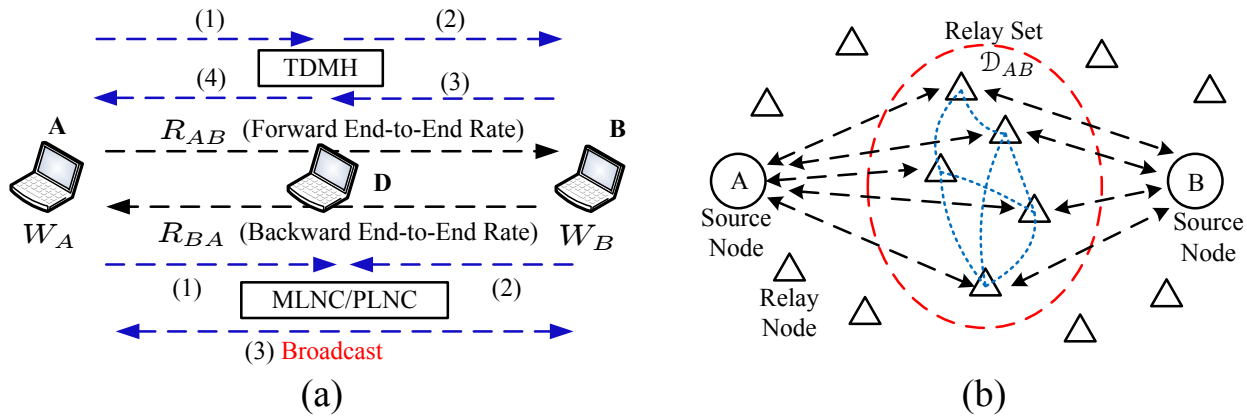}\\
  \caption{(a) Two-way relaying over a single relay: Time-division multihop (TDMH) protocol needs 4 time slots, but wireless network coding only needs 3 time slots. $R_{AB}$ is the end-to-end \emph{forward} rate and $R_{BA}$ is the end-to-end \emph{backward} rate. (b) Two-way relaying over multiple relays: All relay nodes in $\mathcal{D}_{AB}$ are available to collaborate. Thus $\mathcal{D}_{AB}$ becomes a big virtual relay node equipped with $|\mathcal{D}_{AB}|$ antennas, the channels from node A ($\mathcal{D}_{AB}$) to $\mathcal{D}_{AB}$ (A) become a SIMO (MISO) channel.}  \label{Fig:TwoWayRelaySys}
\end{figure}

\begin{figure}[h]
  \centering
  \includegraphics[scale=0.8]{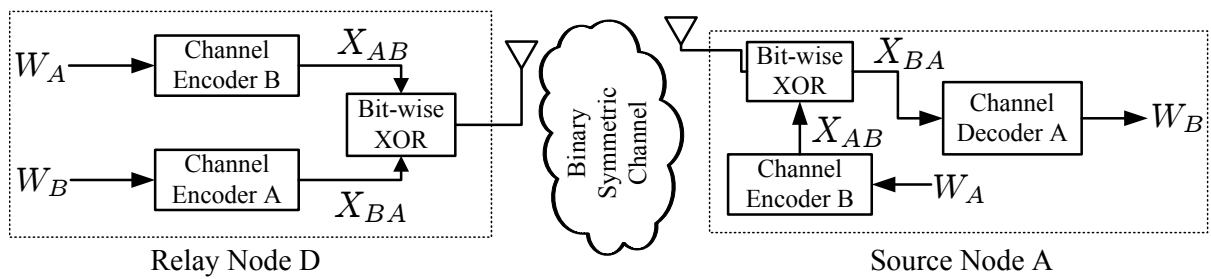}\\
  \caption{PHY-Layer network coding (revised from \cite{FXCHLSS07,FXSS07})}\label{Fig:BinaryPLNC}
\end{figure}

\begin{figure}[h]
  \centering
  \includegraphics[scale=0.8]{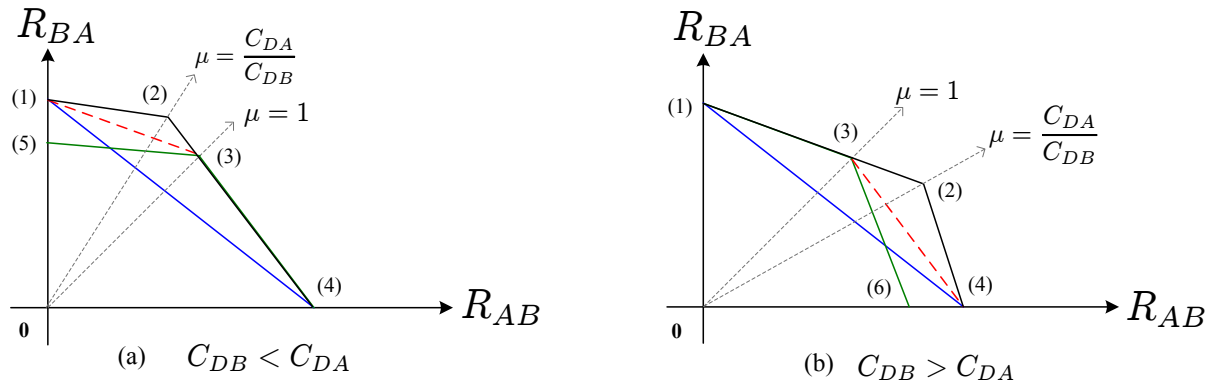}\\
  \caption{Achievable Rate Regions for TDMH, MLNC and PLNC. Verteices (1)-(6) are (0,$\Sigma_{BA}$), ($\Sigma_{ABA}$,$\Sigma_{BAB}$), ($\Sigma_{ABB}$,$\Sigma_{ABB}$), ($\Sigma_{AB}$,0), ($0$,$\Sigma_{BB}$) and ($\Sigma_{AA}$,0), respectively. TDMH is the triangular region with vertices (1), (4), \textbf{0}. MLNC is the quadrilateral region with vertices (1), (3), (4), \textbf{0} in (a) and with vertices (1), (3), (6), \textbf{0} in (b). PLNC is the quadrilateral region with vertices (1), (2), (4), \textbf{0}. Note that MLNC, PLNC and opportunistic MLNC achieve the same rate region with vertices (1), (3), (4), \textbf{0} if $C_{DA}=C_{DB}$.}\label{Fig:AchiRateRegi}
\end{figure}

\begin{figure}[h]
  \centering
  \includegraphics[scale=1]{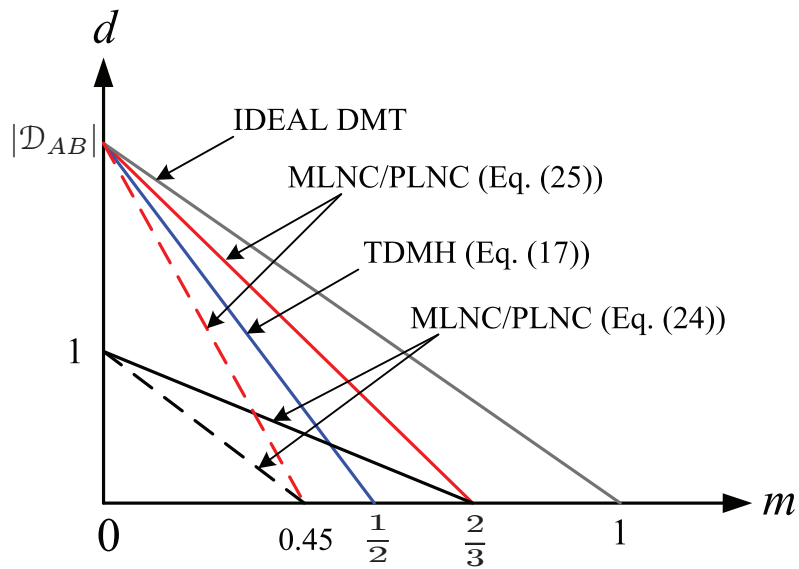}\\
  \caption{Diversity-multiplexing tradeoffs for the three transmission protocols ($|\mathcal{D}_{AB}|>1$ and $\mu=1$). The results of solid lines are the case of optimal time allocation for MLNC and PLNC, \ie $\lambda_f=\lambda_b=0.5$. The results of dashed lines are the case of suboptimal time allocation of MLNC and PLNC, \ie $\lambda_f=0.34$ and $\lambda_b=0.66$.}\label{Fig:DMTradeoff}
\end{figure}

\begin{figure}[h]
  \centering
  \includegraphics[scale=0.6]{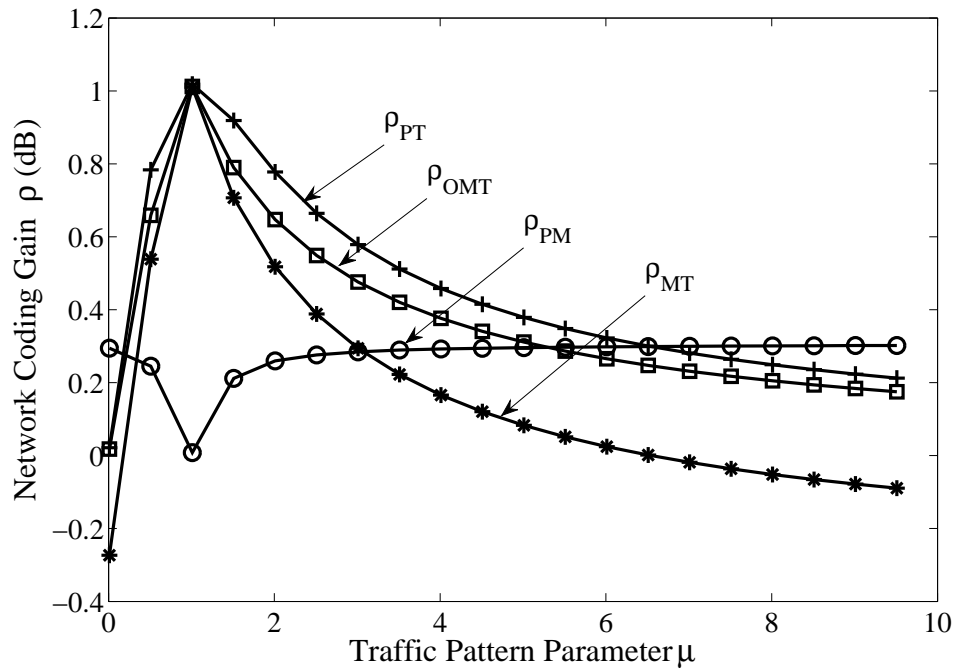}\\
  \caption{Network coding gain $\rho$ vs. traffic pattern parameter $\mu$ for different transmission protocols over a single relay node.}\label{Fig:SimuNCgain}
\end{figure}

\begin{figure}[h]
  \centering
  \includegraphics[scale=0.65]{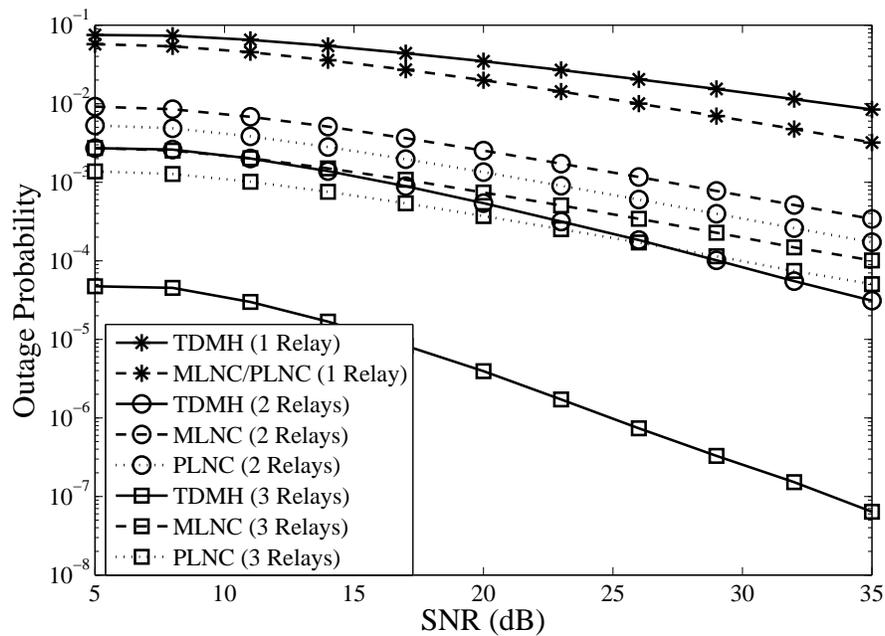}\\
  \caption{Outage probabilities of the TDMH, PLNC and MLNC protocols over multiple relay nodes with collaboration. All relay nodes collaborate to reach receive and transmit MRC for TDMH. For MLNC and PLNC, all relay nodes collaborate to perform receive MRC and then they all broadcast at the same time.}
  \label{Fig:CoopOutageProb01}
\end{figure}

\begin{figure}[h]
  \centering
  \includegraphics[scale=0.65]{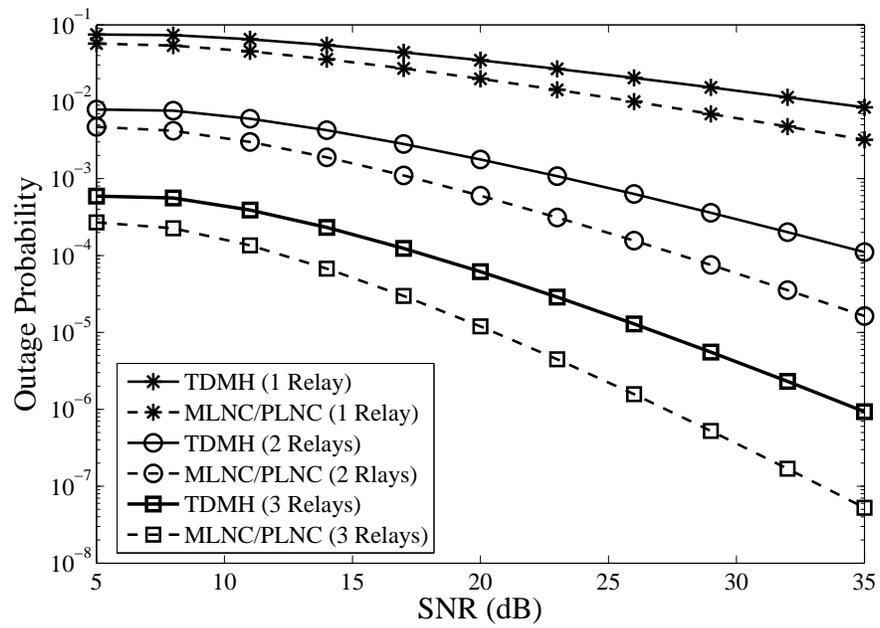}\\
  \caption{Outage probabilities of the TDMH, PLNC and MLNC protocols without relay collaboration. An optimal relay node is selected to receive and transmit/broadcast for all protocols.}
  \label{Fig:NonCoopOutageProb}
\end{figure}

\end{document}